\documentclass[11pt,reqno]{amsart}

\usepackage[utf8]{inputenc}
\usepackage{amsfonts,latexsym,amssymb,amsmath,amsthm,amsrefs,slashed,cancel}
\usepackage{todonotes}
\usepackage{hyperref}
\usepackage{a4wide}
\usepackage{enumerate}
\usepackage{mathrsfs}
\usepackage{nicefrac}
\usepackage{MnSymbol}
\usepackage[all]{xy}

\usepackage[a4paper]{geometry}
\definecolor{blau}{rgb}{0.05,0.2,0.7}
\definecolor{auchblau}{rgb}{0.03,0.3,0.7}

\renewcommand{\Re}{\operatorname{Re}}
\renewcommand{\Im}{\operatorname{Im}}

\DeclareMathOperator{\supp}{supp}
\newcommand{\der}{\mathrm{d}}
\newcommand{\rmi}{\mathrm{i}}

\newcommand{\calO}{\mathcal{O}}
\newcommand{\calE}{\mathcal{E}}
\newcommand{\calB}{\mathcal{B}}

\newcommand{\calK}{\mathcal{K}}
\newcommand{\tr}{\mathrm{Tr}}
\newcommand{\Tr}{\mathrm{Tr}}
\newcommand{\rel}{\mathrm{rel}}	
\newcommand{\dist}{\mathrm{dist}}

\newcommand{\diag}{\mathtt{\Delta}}
\newcommand{\reg}{\mathrm{reg}}

\newcommand{\R}{\mathbb{R}}
\newcommand{\C}{\mathbb{C}}
\newcommand{\N}{\mathbb{N}}
\newcommand{\free}{\textrm{free}}

\newcommand{\loc}{\mathrm{loc}}
\newcommand{\calS}{\mathcal{S}}
\newcommand{\calD}{\mathcal{D}}

\newcommand{\calQ}{\mathcal{Q}}

\newcommand{\dd}{\mathrm{d}}
\newcommand{\ren}{\mathrm{ren}}

\newtheorem{theorem}{Theorem}[section]
\newtheorem{definition}[theorem]{Definition}

\newtheorem{lemma}[theorem]{Lemma}

\newtheorem{proposition}[theorem]{Proposition}

\newtheorem{rem}[theorem]{Remark}

\title[Casimir Forces]{A mathematical analysis of Casimir interactions I: The scalar field}

\author[Y.-L. Fang]{Yan-Long Fang}
\address{School of Mathematics,  University of Leeds,  Leeds , Yorkshire, LS2 9JT,
UK} \email{Y.L.Fang@leeds.ac.uk}
\thanks{Supported by Leverhulme grant RPG-2017-329}

\author[A. Strohmaier]{Alexander Strohmaier}
\address{School of Mathematics,  University of Leeds,  Leeds , Yorkshire, LS2 9JT,
UK} \email{a.strohmaier@leeds.ac.uk}

\begin{document}

\begin{abstract} Starting from the construction of the free quantum scalar field of mass $m\geq 0$ we give mathematically precise and rigorous versions of three different approaches to computing the Casimir forces between compact obstacles. We then prove that they are equivalent.
\end{abstract}

\maketitle


\section{Introduction}
\label{Introduction}
Casimir interactions are forces between objects such as perfect conductors. They can be either understood as quantum fluctuations of the vacuum or as the total effect of van-der-Waals forces.
Hendrik Casimir predicted and computed this effect in the special case of two planar conductors in 1948 using an infinite mode summation (\cite{Casimir}). This force was measured experimentally by Sparnaay 
about 10 years later \cite{sparnaay1958measurements}. Since then also more precise measurements have been performed
with good agreement to the theoretical prediction of Casimir(\cite{Bressi02,chan2001quantum,Ederth00,lamoreaux1997demonstration}). Other geometric situations, such as for example the Casimir force
between a sphere and a plane, were also considered in precision experiments \cite{mohideen1998precision, roy1999improved}. 

The classical way to compute Casimir forces, and indeed the way it was done by Casimir himself, is by performing a zeta function regularisation of the vacuum energy. This has been carried out for a number of particular geometric situations (see \cite{bordag2001,bordag2009advances,elizalde1989expressions,elizalde1990heat,kirsten2002} and references therein). Since this method requires knowledge of the spectrum of the Laplace operator in order to perform the analytic
continuation it has long been a very difficult problem to compute the Casimir force in a generic geometric situation even from a non-rigorous  point of view.
Already, it has been realised by quantum field theorists (see e.g. \cite{BrownMaclay,dzyaloshinskii1961general,DC1979,kay1979casimir,scharf1992casimir}) that the Casimir force can also be understood by considering the renormalised stress energy tensor
of the electromagnetic field. This tensor is defined by comparing the induced vacuum states of the quantum field with boundary conditions and the free theory. 
Once the renormalised stress energy tensor is mathematically defined, the computation of the Casimir energy density becomes a problem of spectral geometry (see e.g. \cite{fulling2007vacuum}) and numerical analysis.
The renormalised stress energy tensor and its relation to the Casimir effect can be understood at the level of rigour of axiomatic algebraic quantum field theory. Currently this is still the subject of ongoing research in mathematical physics (see e.g. \cite{dappiaggi2014casimir})
in particular when it comes to the effect in curved backgrounds and fields that are not scalar, or in situations when the objects move.

Recently progress was made in the non-rigorous numerical computation of Casimir forces between objects (see for example \cite{johnson2011numerical,reid2009efficient,rodriguez2011casimir}). This approach uses a formalism that relates the Casimir energy to a determinant computed from boundary layer operators. Such determinant formulae result in finite quantities that do not require further regularisation and have been obtained and justified in the physics literature \cite{EGJK2017,EGJK2007,EGJK2006,EGJK2008,emig2008casimir, kenneth06, kenneth08, milton08,EGJK2009,renne71}.
This has not only resulted in more efficient numerical algorithms but also in various asymptotic formulae for the Casimir forces for large and small separations. 
Many of the justifications and derivations of these formulae are based on physics considerations of macroscopic properties of matter or of van der Waals forces. As such they often involve ill defined path integrals. From a mathematical point of view considerations that link the determinant formulae to the spectral approach initially taken by Casimir have been largely formal computations or involve ad-hoc cut-offs and regularisation procedures. We note that in the Appendix of \cite{kenneth08} it is proved correctly that the Fredholm determinant in the final formula of the Casimir energy is well defined. Only recently a mathematical justification of these formulae was given in \cite{RTF}, relating it to the trace of a linear combination of powers of the Laplace operator. The determinant formulae are directly related to the multi-reflection expansion of Balian and Duplantier \cite{balian78} that also yields a finite Casimir energy. We mention \cite{kirstenlee18} in the mathematical literature where the Casimir energy of a piston configuration is expressed in terms of the zeta regularised Fredholm determinant of the Dirichlet to Neumann operator.

Summarising, we list several ways to compute the Casimir force acting on a compact object that have been proposed and carried out:
\begin{itemize}
 \item[(1)] Using a total energy obtained in some way by regularising the spectrally defined zeta function. This can be done either directly or by first considering a compact problem in a box and then taking the adiabatic limit. 
  \item[(2)] By integrating the renormalised stress energy tensor around any surface enclosing the object.
 \item[(3)] Using  formulae for the energy in terms of a determinant of boundary layer operator.
 \end{itemize}
The list is non-exhaustive and other methods exist, such as for example the worldline approach (see e.g. \cite{Gies07} and references). Here we will restrict ourselves to the listed methods.

The aim of the present paper is to establish, in the case of finitely many compact objects, the precise mathematical meaning of each of the listed methods for the case of the scalar field and then prove that they give the same answer for the force (not necessarily for the energy). The main tool to achieve this will be the {\sl relative stress energy tensor}. This tensor mimics the definition of the relative trace of \cite{RTF} and seems to not have been defined or studied previously in the literature. Note that the renormalised stress tensor becomes unbounded and non-integrable \cite{Candelas1982,DC1979} near the boundaries of objects and this makes it unsuitable to compute the total energy component from the energy density $T_{00}$. In contrast, the relative stress energy tensor is smooth up to the obstacles and is much more regular when considering boundary variations. This relative stress energy tensor does not satisfy Dirichlet or Neumann boundary conditions and therefore integration by parts involves boundary contributions.
The relative energy density can be defined entirely in terms of functional calculus of the Laplace operator. This relative energy density has been introduced in \cite{RTF}. It was shown to be  integrable and its integral can be interpreted as the trace of a certain operator. The main result of \cite{RTF} states that this trace can be expressed as the determinant of an operator constructed from boundary layer operators, thus providing a rigorous justification of the method (3) linking it with method (1). To show the equivalence of methods (2) and (3) we must provide a formula for the variation of the relative energy when one of the objects is moved. To compute this variation we prove and use a special case of the Hadamard variation formula (\cite{GarabedianSchiffer,Hadamard,Peetre}) adapted to the non-compact setting. We then show that as a consequence of this formula that the variation of the total energy equals the surface integral of the spatial components of the relative stress energy tensor (see Theorem \ref{Equivalent thm}). This surface integral is also equal to the surface integral over the 
renormalised tensor (see Theorem \ref{Equivalent thm 1}). We will now give a more precise formulation of our main result.
\\
We consider $d$-dimensional Euclidean space with $d \geq 2$.
Let $\calO$ be a bounded open subset of $\R^d$ with smooth boundary such that the complement $\mathcal{E} = \R^d \setminus \overline{\calO}$ is connected.
The domain $\calO$ will be assumed to consist of $N$ many connected components $\calO_1,\ldots,\calO_N$. 
The space $X = \R^d \setminus \partial \calO$ therefore consists of the $N+1$-many connected components
$\calO_1,\ldots,\calO_N,\mathcal{E}$. 
We think of $\calO$ as obstacles placed in $\R^d$, and $\mathcal{E}$ then corresponds to the exterior region of these obstacles.
The set $X \subset \R^d$ consists of the interior and the exterior of the obstacles, separated by $\partial \calO$. See Figure \ref{threeobstacles} for an example with three obstacles.

\begin{figure}[h]
	\centering
	\includegraphics[clip, width=5cm]{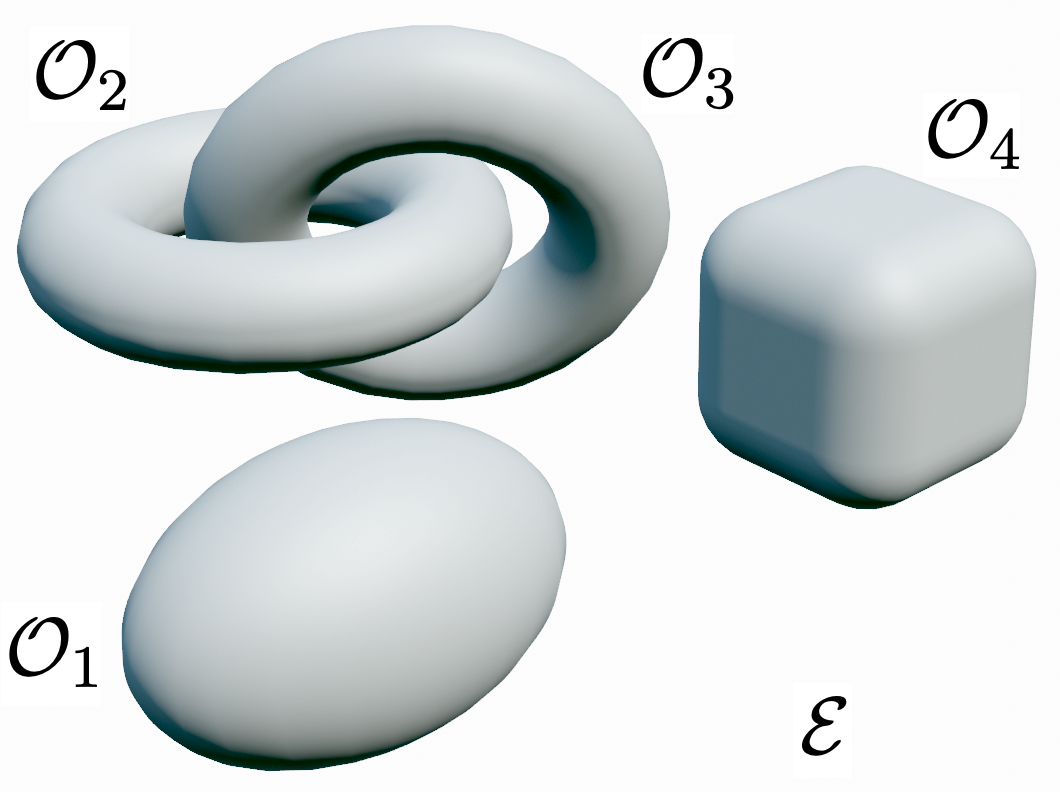}
	\caption{Configuration of four obstacles $\calO_1,\calO_2,\calO_3,\calO_4$ with complement $\calE$ in a Euclidean space}
	\label{threeobstacles}
\end{figure}

For the free scalar field of mass $m \geq 0$ let $T_\ren$ be the renormalised stress energy tensor as defined in Section \ref{The renormalised stress-energy tensor}. In QFT terms this stress energy tensor is equivalent to the usual stress energy tensor obtained from canonical quantisation when normal ordering is used with respect to the free vacuum state.
This is a smooth symmetric two-tensor away from $\partial \mathcal{O}$, but it is singular at $\partial \mathcal{O}$ and the integral of $T_{00}$ over $\mathcal{E}$ does not converge. 
Let $T_\rel$ be the relative stress energy tensor as given in Definition \ref{relstress}.
The relative total energy $E_\rel$ is defined
as the integral of $(T_\rel)_{00}(x)$ over $\R^d$ which can be shown to exist and to be equal the trace of a certain combination of operators. In case $m>0$ the regularised energy $E_\reg$ is defined in Section \ref{zetaregen}, Definition \ref{regenergy} via zeta function regularisation.
We would like to compute the force on an object $\mathcal{O}_j$ due to the presence of the other objects.
Approach (2) is to directly compute
$$
 F_i = \int_\Sigma (T_\ren)_{ik} n^k \der\sigma
$$
where $T_\ren$ is the renormalised stress energy tensor given in Definition \ref{def of Tren} and $\Sigma$ is any smooth surface enclosing $\mathcal{O}_j$ (i.e. homologous to $\partial\mathcal{O}_j$ in $\overline{\calE}$), $n^k$ is the exterior unit normal vector field, and $\der\sigma$ is the surface measure on $\Sigma$. If energy conservation holds, i.e. no energy is radiated off to infinity, one expects this force to be the directional derivative of the total energy when $\mathcal{O}_j$ is moved rigidly. Thus, let $Z$ be a constant vector field on $\R^d$. Let $Y$ be a vector field that equals $Z$ near $\mathcal{O}_j$ and vanishes near $\mathcal{O}_l$ when $l \not= j$. The vector field generates a flow $\Phi_\epsilon$ that, for $\epsilon$ near zero, moves the object $\mathcal{O}_j$ rigidly and we end up with a configuration that depends on the parameter $\epsilon$. The total energies $E_\rel$ and $E_\reg$ then also depend on $\epsilon$ in that way and become functions of $\epsilon$ in a small interval around zero.  The change of these energies with respect to the flow can be interpreted as the change of energy needed to move object $\calO_j$ rigidly relative to the other objects in the direction of $Z$.
Our main result may be stated as follows.

\begin{theorem}
	\label{main thm}
 The relative energy $E_\rel(\epsilon)$ is differentiable in $\epsilon$ at $\epsilon=0$ and its derivative equals $F_i Z^i$, where
 \begin{equation}
 \label{main thm eqn1}
  F_i = \int_\Sigma (T_{\rel})_{ik} n^k \der \sigma= \int_\Sigma (T_{\ren})_{ik} n^k \der \sigma.
 \end{equation}
 Moreover,
 \begin{equation}
 \label{main thm eqn2}
  E_\rel = \frac{1}{2\pi} \int_{m}^\infty \frac{\omega}{\sqrt{\omega^2-m^2}} \log \det Q_\rel(\omega) \der \omega,
 \end{equation}
 where $Q_\rel(\omega) = Q_{\rmi \omega} (\tilde Q_{\rmi \omega})^{-1}$ is defined in detail in Section \ref{The relative trace-formula and the Casimir energy} and constructed out of boundary layer operators for the Laplacian. If $m>0$ then $E_\rel(\epsilon) - E_{\reg}(\epsilon)$ is constant near $\epsilon=0$.
\end{theorem}
 
We note that this mathematical theorem simply shows that all these proposed computational methods give the same Casimir interactions in the case of separated rigid bodies. The statement does not say anything about the actual origin of the Casimir force or its existence, which needs to be determined from experiments or physics considerations. 
There is however a strong argument for the expression
$$
  F_i = \int_\Sigma (T_{\rel})_{ik}  n^k \der \sigma
 $$
to be a directly relevant physical quantity. Our point of view is that the stress energy tensor does not have an absolute meaning in this context, but rather is used to compare two vacuum states (normal ordering depends on a comparison state). If we would like to know the effect for the rigid objects $\mathcal{O}_2,\ldots \mathcal{O}_N$ on the rigid object $\mathcal{O}_1$ the states to compare are not the ground state with Dirichlet conditions and the free ground state. It is rather the vacuum states obtained from the Laplacian with Dirichlet conditions imposed on $\mathcal{O}_1$ alone and with Dirichlet conditions on all the objects. The comparison of these two states yields a stress energy tensor that is completely regular near $\mathcal{O}_1$, and the computation of the force based on this tensor leads directly to the above formula without regularisation.
 
The paper is organised as follows. In Section \ref{Sec:2} we review the rigorous construction of the free scalar field of mass $m\geq 0$ in the presence of boundaries and show how this leads to a natural definition of the renormalised stress energy tensor, which is given in Section \ref{The renormalised stress-energy tensor}. We also review its most important properties and express it in terms of spectral quantities for the Laplace operator. Section \ref{The relative trace-formula and the Casimir energy} introduces the relative setting and gives the definition of the relative stress energy tensor and its basic properties. Some norm estimates on the relative resolvent are given in Section \ref{Estimates on the relative resolvent}, which provides mathematical justifications for later proofs. In Section \ref{Had:Sec} we prove a Hadamard variation formula and compute the variation of the relative energy to establish the first part of the main theorem. In Section \ref{zetaregen} we show that for $m>0$ the renormalised version of the zeta function has a meromorphic continuation and can thus be used to define the regularised energy. This section also contains a proof that variations of the regularised energy and the relative energy coincide. To illustrate the method and relate it to the classical computations we treat the easier example of the one-dimensional Casimir energy explicitly in Appendix \ref{1-D relative setting}. This example also illustrates that a divergence term for the time-component of the renormalised stress energy tensor that is normally neglected needs to be taken into account to obtained the correct result (see Remark \ref{divremark}).
 
In a follow-up paper we will establish a similar theorem for the electromagnetic field. We note here that the stress energy for the electromagnetic field is quite different from the scalar field and there are additional complications such as zero modes (\cite{MR4162947}) that are absent for the scalar field. Moreover, the boundary conditions for the electromagnetic field are slightly more complicated, and cannot be reduced to Dirichlet boundary conditions.
We therefore decided to not attempt a unified treatment which would obscure the result by additional notations.

Our approach is expected to carry over to other boundary conditions such a Neumann, mixed Dirichlet-Neumann, or Robin boundary conditions with the single layer operators replaced by the appropriate layer operators. As in the electromagnetic case additional technical problems need to be overcome in these cases due to the possible appearance of zero modes and singularities of the Dirichlet-to-Neumann map at zero.

\subsection{Notations}
\label{Notations}
Let $M \subset \R^d$ be an open subset.
By the Schwartz kernel theorem continuous linear operators $A: C^\infty_0(M) \to \mathcal{D}'(M)$ are in one to one correspondence to distributions in $\mathcal{D}'(M \times M)$, i.e. for every such $A$ there exists a unique Schwartz kernel in  $\mathcal{D}'(M \times M)$. In this paper the Schwartz kernel of $A$ will be denoted by $\breve A$.

\section{Scalar quantum field theory with Dirichlet boundary conditions} \label{Sec:2}

Let $-\Delta$ be the (positive) Dirichlet Laplacian imposed on the codimension one submanifold $\partial \calO$. By definition this is the unbounded self-adjoint operator defined  on the Hilbert space $L^2(\R^d)$ associated with the Dirichlet quadratic form $q_D(f,f) = \int_X \| \nabla f \|^2 \dd x $ with form-domain being the Sobolev space $H^1_0(X)$.

As a consequence of elliptic regularity (\cite{ShubinPseudos}*{Section 7.2})  we have for the domain of $\mathrm{dom}((-\Delta)^{\frac{s}{2}})$ equipped with its graph norm for any $s \in 2 \N_0$ the continuous inclusions
$$
 H^{s}_{\mathrm{comp}}(X) \subset \mathrm{dom}((-\Delta)^{\frac{s}{2}}) \subset H^{s}_{\loc}(X).
$$
One can use complex interpolation (\cite{PseudoTaylor}*{\S 4 and Theorem 4.2}) to extend this to any $s \geq 0$.
Here $H^{s}_{\mathrm{comp}}(X)$ denotes the space of functions in $H^{s}(X)$ with compact support in $X$. 

The Hilbert space $L^2(\R^d)$ then decomposes into a direct sum $$L^2(\R^d) = L^2(\calO_1) \oplus \ldots \oplus  L^2(\calO_N) \oplus L^2(\mathcal{E})$$ and each subspace is an invariant subspace for $-\Delta$  in the sense that any bounded function of the operator as defined by spectral calculus will leave these subspaces invariant.
The restriction of $-\Delta$ to $L^2(\calO_j)$ is the Dirichlet Laplacian on the interior of $\calO_j$ and therefore has compact resolvent. The restriction of $-\Delta$ to $L^2(\mathcal{E})$ has purely absolutely continuous spectrum $[0,\infty)$.
By comparison we also have the free Laplacian $-\Delta_\free$ on $\R^d$ which corresponds to the case $\calO=\emptyset$. Throughout the paper we fix a mass parameter $m \geq 0$.
\begin{definition}
The relativistic Hamiltonian $H$ is defined to be the self-adjoint operator $H = (-\Delta + m^2)^{\frac{1}{2}}$.
\end{definition}

The space-time $M$ we consider is the Lorentzian spacetime $\R \times X$ with Minkowski metric. The forward and backward fundamental solutions $G_{+/-} : C^\infty_0(M) \to C^\infty(M)$ of the Klein-Gordon operator $\Box + m^2$ with Dirichlet boundary conditions are given by

\begin{gather*}
 (G_+ f)(t,x)=\int \left(\theta(t-t') H^{-1} \sin(H(t-t')) f(t',\cdot)\right)(x) \dd t',\\
 (G_- f)(t,x)=\int \left(\theta(t'-t) H^{-1} \sin(H(t'-t)) f(t',\cdot)\right)(x) \dd t',
\end{gather*}
where $\theta$ is the Heaviside step function. As usual in canonical quantisation one considers the difference $G = G_+ - G_-$ given by
\begin{gather*}
 (G f)(t,x)=\int \left( H^{-1} \sin(H(t-t')) f(t',\cdot)\right)(x) \dd t'.
\end{gather*}
Here $H^{-1} \sin(H(t-t'))$ is defined by spectral calculus. Since the function $| x^{-1} \sin(x t) |$ is bounded by $|t|$ the operator
 $H^{-1} \sin(H(t-t'))$ defines for any $s\geq 0$ a bounded map from $H^s_{\mathrm{comp}}(X)$ to $\mathrm{dom}(-\Delta)^s \subset
 H^s(\calO_1) \oplus \ldots H^s(\calO_N) \oplus H^s(\calE)$. Here the inclusion of the domain in the Sobolev spaces follows for $s \in 2 \N_0$ from elliptic regularity up to the boundary (\cite{McLean2000}*{Theorem 4.18}) and for general $s \geq 0$ by interpolation.
In particular, this means that $H^{-1} \sin(H(t-t'))$ has a distributional integral kernel.
We can define a symplectic structure $\sigma$ on $\mathcal{W}=C^\infty_0(M)/\left( (\Box+m^2)C^\infty_0(M) \right)$
by
$$
 \sigma([f],[g])= ( f, G g ).
$$
This induces a symplectic structure on $G C^\infty_0(M)$ that is well known to coincide with the standard symplectic structure on the space of solutions. Indeed, if we define $u = G f$ and $v = G g$ for $f,g \in C^\infty_0(M)$ then $u$ and $v$ solve the Klein-Gordon equation with Dirichlet boundary conditions and
$$
 ( f, G g ) = \int_X (\partial_t u) v - (\partial_t v) u\dd x.
$$
In this equality the right hand side is independent of $t$.

\subsection{Field algebra and the vacuum state}

The field algebra of the real Klein-Gordon field is then the (unbounded) CCR $*$-algebra of this symplectic space. 
Instead of using the symplectic space $\mathcal{W}$ one can describe this algebra $\mathcal{A}$ directly as the complex unital $*$-algebra
generated by symbols $\phi(f)$ with $f \in C^\infty_0(M,\R)$ satisfying the relation
\begin{align*}
&f \to \phi(f) \quad \textrm{is real linear},\\
&[\phi(f_1),\phi(f_2)] = -\rmi ( f_1, G f_2 ) \mathbf{1}\\
&\phi(f)^*=\phi(f),\\
&\phi((\Box + m^2) f)=0.
\end{align*}
Physical states of this quantum system are states on this $*$-algebra.
The construction and physical interpretations of such states usually relies on
a Fock representation of $\mathcal{A}$. This representation is chosen on physical
grounds as a positive energy representation. 

We briefly explain this now.
The group of time translations $(T_s f)(t,x)=f(t-s,x)$ commutes with $\Box + m^2$
and $G$ and therefore $\alpha_t(\phi(f)):=\phi(T_t f)$ defines a group of $*$-automorphisms
of $\mathcal{A}$. If a state $\omega:\mathcal{A} \to \C $ is invariant then this group lifts to a group of unitary
transformations $U(t)$ on the GNS-Hilbert space which is uniquely determined by
\begin{align*}
 &\pi(\alpha_t(a))=U_{t} \pi(a) U_{-t},\\
 &U_t \Omega = \Omega.
\end{align*}
We say that $\pi$ is a positive energy representation if this group is strongly continuous
and its infinitesimal generator has non-negative spectrum. 

We will focus in this paper on the quasi-free ground state. This means that the state is completely determined by its two point distribution
$$\omega_2 \in \mathcal{D}'(M \times M),\quad \omega_2(f_1 \otimes f_2) = \omega \left( \phi(f_1) \phi(f_2) \right)$$ which is given explicitly as
$$
  \omega_2(f_1 \otimes f_2) =\int_{\R \times \R} ( f_1(t,\cdot), \frac{1}{2}H^{-1}e^{-\rmi H (t-t')}  f_2(t',\cdot) ) \:\der t \; \der t',
$$
i.e. $\omega_2$ is the integral kernel of the operator $H^{-1} e^{-\rmi H(t-t')}$. In case $m=0$ the spectrum of $H$ contains zero and  this expression needs to be interpreted in the sense of quadratic forms. Namely,
it follows from general resolvent expansions
(for example \cite{MR4162947}*{Theorem 1.5, 1.6 and 1.7})
 that $C^\infty_0(X)$ is contained in the domain of the operator $H^{-\frac{1}{2}}$. This follows from the formula
 $$
  \langle H^{-\frac{1}{2}} f, H^{-\frac{1}{2}} f \rangle  = \frac{2}{\pi}  \int_0^\infty \langle (-\Delta + \lambda^2)^{-1} f, f \rangle \der \lambda.
 $$ In particular, the space of test functions $C^\infty_0(X)$
is contained in the form domain of $H^{-1}$ and therefore, by the Schwartz kernel theorem, the operator $H^{-1}$ has a distributional kernel in $\mathcal{D}'(X \times X)$. This is of course also the case for general $m>0$. We will denote the integral kernel of $H^{-1}$
by $\breve H^{-1}$, mildly abusing notation.

 One can check directly that $\omega_2$ defines a positive energy representation. Instead of using $-\Delta$ we could also have used $-\Delta_\free$, the free Laplace operator. This also defines a positive energy state on the free algebra of observables $\mathcal{A}_\free$ which we will denote by $\omega_{\free}$, and similarly we use the notation $H_\free= (-\Delta_\free + m^2)^\frac{1}{2}$. There states can be compared by restricting them to certain subalgebras that are contained in both the algebra of observables and the free algebra of observables. For example if $U$ is contained in a double cone in $M$ then $\mathcal{A}(U)$ which is generated by $\phi(f), \mathrm{supp}(f) \subset U$ can be thought of as a subset of both $\mathcal{A}$ and $\mathcal{A}
 _\free$ and therefore both states can be restricted to this algebra.

\subsection{The renormalised stress-energy tensor}
\label{The renormalised stress-energy tensor}
The classical stress energy tensor of the Klein-Gordon field for a smooth real-valued solution $u$ is given by
$$
 T^{cl}(u) =\der u \otimes \der u - \frac{1}{2} g(\der u, \der u) g - \frac{1}{2} g\, m^2 u^2.
$$
This is a symmetric $2$-tensor on $M$ and one can easily show, using the Klein-Gordon equation, that it is divergence-free.
Here $g$ is the Minkowski metric on $M$ with signature $(-1,1,\ldots,1)$. The Euclidean metric on $\R^d$ will be denoted by $h$. The components $T_{ij}^{cl}$ of the stress energy tensor are the restrictions 
to the diagonal of the functions $Q^{cl}_{ij}(x,y)$ defined on $M \times M$ by
$$
 Q^{cl}_{ij}(x,y)= (\nabla_i u)(x) \otimes (\nabla_j u)(y) - \frac{1}{2} g_{ij} (\nabla_k u)(x) \cdot \nabla^k u(y)  - \frac{1}{2} g_{ij} m^2 u(x) u(y).
$$
The quantum field theory counterpart of $Q$ can be written in the field algebra
as a field-algebra-valued bilinear form in the test functions $f_1,f_2 \in C^\infty(M)$ as 
$$
 Q_{ij}(f_1 \otimes f_2) = \phi( \nabla_i f_1) \phi(\nabla_j f_2) -  \frac{1}{2}  \phi(\nabla_k f_1) \phi(\nabla^k f_2)  g_{ij} -  
 \frac{1}{2} g_{ij} m^2 \phi(f_1) \phi( f_2)   .
$$
The expectation value of $Q_{ij}$ with respect to the state $\omega$ is then given in terms of the two point function $\omega_2$ as
\begin{equation*}
 \omega(Q_{ij})(f_1,f_2) = \omega_2( \nabla_i f_1 \otimes \nabla_j f_2 -  \frac{1}{2}  \nabla_k f_1\otimes  \nabla^k f_2  g_{ij} -  \frac{1}{2} g_{ij} m^2 f_1 \otimes f_2 ).   
\end{equation*}

Let $K(t,x,y)$ be the distributional kernel of $\frac{1}{2}H^{-1}e^{-\rmi Ht}$. Then the distribution $\omega(Q_{ij})$ with respect to the ground state is given by
\begin{equation*}
\left\{
\begin{aligned}
&\omega(Q_{00})(t,t',x,x')=\frac{1}{2}\left(\partial_t \partial_{t'}+\sum_{k=1}^d \partial_{x^k}\partial_{x'^k} +m^2 \right)K(t-t',x,x')\\
&\omega(Q_{0i})(t,t',x,x')=\left(\partial_{t'}\partial_{x^i}\right)K(t-t',x,x')\quad \text{for}\quad 1\le i\le d\\
&\omega(Q_{ij})(t,t',x,x')=\left(\partial_{x'^i}\partial_{x^j}\right)K(t-t',x,x')\quad \text{for}\quad i\ne j \quad \text{and}\quad 1\le i,j \le d\\
&\omega(Q_{ii})(t,t',x,x')=\frac{1}{2}\left(\partial_t \partial_{t'}+\partial_{x^i}\partial_{x'^i}-\sum_{k\ne i}^d \partial_{x^k}\partial_{x'^k} -m^2 \right)K(t-t',x,x')\quad \text{for}\quad 1\le i \le d
\end{aligned}
\right.
\end{equation*}
or in short,
\begin{equation} \label{kernel of omega}
\omega(Q_{ij})(t,t',x,x')=\left[\partial_i'\partial_jK-\frac{1}{2}g_{ij}(\partial_k'\partial^k+m^2)K\right](t-t',x,x').
\end{equation}

The above expressions are formal and make sense only when paired with test functions. We will use such formal notation throughout the paper when there is danger of confusion.
The expectation value of the stress energy tensor would correspond to the restriction of  $\omega(Q_{ij})$ to the diagonal as a distribution. Unfortunately, the distribution $\omega(Q_{ij})$ is singular and cannot be restricted to the diagonal in a straightforward manner. If one is interested in relative quantities only then one can still define the renormalised expectation value of the stress energy tensor between the states. Both states $\omega$ and $\omega_\free$ are positive energy states and therefore satisfy the Hadamard condition (for example \cite{SVW02}*{Theorem 6.3}). By uniqueness of such Hadamard states the difference of the two-point distributions is smooth near the diagonal in $M \times M$. In the present case this can also be seen more directly as follows.
Let $K_\free(t,x,y)$ be the kernel of $\frac{1}{2}H_\free^{-1}e^{-\rmi H_\free t}$. We will consider the difference
$\calK(t,x,y) =K(t,x,y) - K_\free(t,x,y)$.

\begin{theorem}
	\label{smooth K thm}
 The distribution $\calK$ is smooth near the set $\{ (0,x,y) \mid x,y \in X\} \subset \R \times X \times X$. In particular $\breve H^{-1} - \breve H^{-1}_\free$ is smooth in $X \times X$.
\end{theorem}
\begin{proof}
 The distribution $\calK(t,x,y)$ is a solution of the wave equation 
 $$
  \partial_t^2 - \frac{1}{2}\left( \Delta_x + \Delta_y  \right) \calK(t,x,y) =0
 $$
 on $\R \times X \times X$ with initial conditions
 \begin{align*}
 &\calK(0,x,y) = \frac{1}{2} \left( \breve H^{-1}(x,y) - \breve H^{-1}_\free(x,y) \right), \\
 &\partial_t \calK(t,x,y) |_{t=0} = 0.
\end{align*}
By \cite{RTF}*{Lemma 5.1} integral kernel $\kappa_\lambda$ of the resolvent difference $$(-\Delta + \lambda^2)^{-1} - (-\Delta_\free + \lambda^2)^{-1}$$ is smooth and satisfies on any compact subset $U$ of $X \times X$ a $C^k$-norm bound of the form
 \begin{gather} \label{expdecay}
  \| \kappa_{\lambda} \|_{C^k(U)} \leq C_{k,U} |\log(|\lambda|)|^\ell e^{-\delta' \lambda} \textrm{ for all } \lambda >0
 \end{gather}
 for some $\ell \geq 0$ and $\delta'>0$. Therefore the integral representation
 $$
  H^{-1} - H^{-1}_\free = \frac{2}{\pi} \int_m^\infty \frac{\lambda}{\sqrt{\lambda^2 - m^2}} \kappa_{\lambda} \der \lambda
 $$
 converges in the $C^\infty(X \times X)$ topology. Thus the distribution $\breve H^{-1} - \breve H^{-1}_\free$ is smooth in $X \times X$.
 Since the initial conditions are smooth the solution $\calK$ is smooth where it is uniquely determined by the initial data. 
 This is the case in a neighborhood of $t=0$ in $\R \times X \times X$.
\end{proof}

Hence the distribution $\omega(Q_{ij})- \omega_\free(Q_{ij})$ is a smooth function in a neighbourhood of the diagonal $\diag \subset M\times M$.
\begin{definition}
	\label{def of Tren}
 The components $(T_\ren)_{ij}$ of the renormalised stress energy tensor $T_\ren$  are defined to be the restriction to the diagonal of the function $\omega(Q_{ij})- \omega_\free(Q_{ij})$.
\end{definition}

\begin{theorem}
	\label{expression of T}
The renormalised stress energy tensor is symmetric and is given by
\begin{equation}
\label{stress energy tensor}
\left\{
\begin{aligned}
&(T_\ren)_{00}= \frac{1}{2}(\breve{H}-\breve H_\free)|_\diag+\frac{1}{8}\Delta[(\breve H^{-1}-\breve H_\free^{-1})|_\diag]\\
&(T_\ren)_{ij}=\frac{1}{2}[\partial_i\partial_j'(\breve H^{-1}-\breve H_\free^{-1})]|_\diag-\frac{1}{8}h_{ij}\Delta[(\breve H^{-1}-\breve H_\free^{-1})|_\diag]\\
&(T_\ren)_{0i}=0
\end{aligned}
\right. 
\end{equation}
and
\begin{equation*}
\left\{
\begin{aligned}
&(T_\ren)_{00}= \frac{1}{2}(\breve H-\breve H_\free)|_\diag+\frac{1}{8}\Delta[(\breve H^{-1}-\breve H_\free^{-1})|_\diag]\\
&(T_\ren)_{ij}=-\frac{1}{2}[\partial_i\partial_j (\breve H^{-1}-\breve H_\free^{-1})]|_\diag+\frac{1}{4}\partial_i\partial_j[(\breve H^{-1}-\breve H_\free^{-1})|_\diag]-\frac{1}{8}h_{ij}\Delta[(\breve H^{-1}-\breve H_\free^{-1})|_\diag]\\
&(T_\ren)_{0i}=0
\end{aligned}
\right. 
\end{equation*}
for $1\le i,j \le d$. Note that here $\breve H^{-1}$ and $\breve H^{-1}_\free$ are the the integral kernels of $H^{-1}$ and $H^{-1}_\free$ respectively. Moreover, the expression $\breve A|_\diag$ means the restriction of the integral kernel, $\breve A$, to the diagonal (See Section \ref{Notations}).
\end{theorem}
\begin{rem}
	\label{divremark}
The terms of divergence forms in the renormalised stress energy tensor are commonly neglected in the literature, as one may naively think that they have zero contribution when integrating over the whole space. However, this is not the case. As it is not integrable due to the singular behaviour near the boundary, the divergence theorem does not apply in this case. See Appendix \ref{1-D relative setting} for the simplest case. The problem disappears when we work with the relative stress energy tensor given in Definition \ref{relstress}.
\end{rem}
\begin{proof}[Proof of Theorem \ref{expression of T}]
Let $\calK(t,x,x')$ be the kernel of $\frac{1}{2}(H^{-1}e^{-\rmi Ht}-H_\free^{-1}e^{-\rmi H_\free t})$ and $\mathscr{K}(t,t',x,x')=\calK(t-t',x,x')$. We have that
\begin{equation}
\label{T and A}
(T_\ren)_{ij}=\left.\left[\partial_i'\partial_j-\frac{1}{2}g_{ij}(\partial_k'\partial^k+m^2)\right]\mathscr{K}\right|_\diag,
\end{equation}
where $\partial_0=\partial_t$, $\partial'_0=\partial_{t'}$, $\partial_i=\partial_{x^i}$ and $\partial_i'=\partial_{x'^i}$ for $1\le i \le d$. By theorem \ref{smooth K thm} the distribution $\mathscr{K}(t,t',x,x')$ is smooth in a neighbourhood of the diagonal $\diag \subset M\times M$. Moreover, $\calK(t-t',\cdot,\cdot)$ is the kernel of a symmetric operator on $X$ with respect to the real inner product $(\cdot,\cdot)$ and therefore satisfies 
\begin{equation}
\label{symmetry of A}
\calK(t-t',x,x')=\mathscr{K}(t,t',x,x')=\mathscr{K}(t,t',x',x)=\calK(t-t',x',x).
\end{equation}
This implies
\begin{equation}
\label{A eqn 1}
\left\{
\begin{aligned}
&\partial_0\mathscr{K}(t,t',x,x')=-\partial_0'\mathscr{K}(t,t',x,x')\\
&\partial_i\mathscr{K}(t,t',x,x')=\partial_i'\mathscr{K}(t,t',x',x)\\
&\partial_i\partial_j\mathscr{K}(t,t',x,x')=\partial_i'\partial_j'\mathscr{K}(t,t',x',x)\\
&\partial_i\partial_j'\mathscr{K}(t,t',x,x')=\partial_i'\partial_j\mathscr{K}(t,t',x',x)
\end{aligned}
\quad \text{for} \quad 1\le i,j \le d\right..
\end{equation}
Using product rules, we have
\begin{equation}
\label{product rules}
\partial_i(\mathscr{K}|_\diag)=(\partial_i\mathscr{K})|_\diag + (\partial_i'\mathscr{K})|_\diag ,
\end{equation}
which gives
\begin{equation*}
\partial_i\partial_j(\mathscr{K}|_\diag)=(\partial_i\partial_j \mathscr{K})|_\diag+(\partial_i\partial_j' \mathscr{K})|_\diag +(\partial_j\partial_i'\mathscr{K})|_\diag +(\partial_i'\partial_j'\mathscr{K})|_\diag .
\end{equation*}
That is
\begin{equation}
\label{A eqn 2}
 (\partial_i\partial_j' \mathscr{K})|_\diag=\partial_i\partial_j(\mathscr{K}|_\diag)-(\partial_i\partial_j \mathscr{K})|_\diag - (\partial_j\partial_i'\mathscr{K})|_\diag -(\partial_i'\partial_j'\mathscr{K})|_\diag .
\end{equation}
Applying equation \eqref{A eqn 1} to \eqref{A eqn 2}, we have
\begin{equation}
\label{symmetric of T1}
\left\{
\begin{aligned}
&(\partial_0\partial_0' \mathscr{K})|_\diag=-(\partial_0\partial_0 \mathscr{K})|_\diag = (\breve H-\breve H_\free)|_\diag  \\
&(\partial_0\partial_i' \mathscr{K})|_\diag=-(\partial_0'\partial_i'\mathscr{K})|_\diag=(\partial_0\partial_i\mathscr{K})|_\diag=-(\partial_i\partial_0' \mathscr{K})|_\diag\quad \text{for} \quad 1\le i \le d\\
&(\partial_i\partial_j' \mathscr{K})|_\diag=\frac{1}{2}\partial_i\partial_j(\mathscr{K}|_\diag)-(\partial_i\partial_j \mathscr{K})|_\diag\quad \text{for} \quad 1\le i,j \le d
\end{aligned}
\right. .
\end{equation}
From equations \eqref{A eqn 1} and \eqref{product rules} we have
\begin{equation}
\label{symmetric of T2}
2\partial_i\partial_0\mathscr{K}|_\diag=\partial_i\partial_0\mathscr{K}|_\diag+\partial_i'\partial_0\mathscr{K}|_\diag=\partial_i(\partial_0\mathscr{K}|_\diag)=0\quad \text{for} \quad 1\le i \le d.
\end{equation}
In other words, $(T_\ren)_{0i}=(T_\ren)_{i0}=0$ for $1\le i \le d$. Hence equations \eqref{T and A}, \eqref{symmetric of T1} and \eqref{symmetric of T2} show that $(T_\ren)_{ij}$ is symmetric tensor on $M$. Moreover, 
\begin{equation*}
(T_\ren)_{ij}=\frac{1}{2}\partial_i\partial_j(\mathscr{K}|_\diag)-(\partial_i\partial_j \mathscr{K})|_\diag\quad \text{for} \quad 1\le i\ne j \le d\\
\end{equation*}
and
\begin{multline*}
(T_\ren)_{ii}=\left.\left[\partial_i'\partial_i-\frac{1}{2}g_{ii}(\partial_k'\partial^k+m^2)\right]\mathscr{K}\right|_\diag\\
=\frac{1}{2}\partial_i\partial_i(\mathscr{K}|_\diag)-(\partial_i\partial_i \mathscr{K})|_\diag-\frac{1}{2}g_{ii}\left(\frac{1}{2}\partial_k\partial^k(\mathscr{K}|_\diag)-(\partial_k\partial^k \mathscr{K}-m^2\mathscr{K})|_\diag\right) .
\end{multline*}
Since $\partial_k\partial^k \mathscr{K}-m^2\mathscr{K}=0$, we have
\begin{equation*}
(T_\ren)_{ii}=\frac{1}{2}\partial_i\partial_i(\mathscr{K}|_\diag)-(\partial_i\partial_i \mathscr{K})|_\diag-\frac{1}{4}g_{ii}\partial_k\partial^k(\mathscr{K}|_\diag) .
\end{equation*}
When $i=0$,
\begin{equation*}
(T_\ren)_{00}=\frac{1}{2}(\breve H-\breve H_\free)|_\diag+\frac{1}{8}\Delta[(\breve H^{-1}-\breve H_\free^{-1})|_\diag] .
\end{equation*}
Also, we have
\begin{equation*}
(\partial_i\partial_j \mathscr{K})|_\diag =\frac{1}{2}[\partial_i\partial_j(\breve H^{-1}-\breve H_\free^{-1})]|_\diag \quad \text{for} \quad 1\le i, j \le d ,
\end{equation*}
which yields the expressions for the renormalised stress energy tensor.
\end{proof}

\begin{theorem}
\label{DF of T}
 The renormalised stress energy tensor is divergence-free and independent of $t$.
\end{theorem}
\begin{proof}
Let $\calK(t,x,y)$ and $\mathscr{K}(t,t',x,y)$ be the same as the in the previous theorem. Recall that the shorthand expression of \eqref{kernel of omega} is given by
\begin{equation*}
\omega(Q_{ij})(t,t',x,x')=\left[\partial_i'\partial_j\mathscr{K}-\frac{1}{2}g_{ij}(\partial_k'\partial^k+m^2)\mathscr{K}\right](t,t',x,x').
\end{equation*}
Then we have
\begin{equation*}
(T_\ren)_{ij}(t,x)=\omega(Q_{ij})|_\diag=\left.\left[\partial_i'\partial_j\mathscr{K}-\frac{1}{2}g_{ij}(\partial_k'\partial^k+m^2)\mathscr{K}\right]\right|_\diag.
\end{equation*}
Now we use product rules to get 
\begin{equation*}
\partial_{i_0}\left(\partial_{i_1}\cdots \partial_{i_m}\mathscr{K}|_\diag\right)=\left(\partial_{i_0}\partial_{i_1}\cdots \partial_{i_m}\mathscr{K}\right)|_\diag+ \left(\partial_{i_0}' \partial_{i_1}\cdots \partial_{i_m}\mathscr{K}\right)|_\diag .
\end{equation*}
In particular, we have
\begin{equation*}
\partial_0\left(\partial_{i_1}\cdots \partial_{i_m}\mathscr{K}|_\diag\right)=0 .
\end{equation*}
Hence, one has
\begin{equation*}
\partial_0 (T_\ren)_{ij}(t,x)=\partial_0\left\{\left.\left[\partial_i'\partial_j\mathscr{K}-\frac{1}{2}g_{ij}(\partial_k'\partial^k+m^2)\mathscr{K}\right]\right|_\diag\right\}=0 ,
\end{equation*}
which means $(T_\ren)_{ij}$ is independent of time.

We use the same trick to show that
\begin{multline*}
\partial^i (T_\ren)_{ij}(t,x)=\partial^i\left\{\left.\left[\partial_i'\partial_j\mathscr{K}-\frac{1}{2}g_{ij}(\partial_k'\partial^k+m^2)\mathscr{K}\right]\right|_\diag\right\}\\
=\left.\left(\partial^i\partial_i'\partial_j\mathscr{K}+\partial'^i\partial_i'\partial_j\mathscr{K}-\frac{1}{2}\partial_j\partial_k'\partial^k\mathscr{K}-\frac{1}{2}\partial_j'\partial_k'\partial^k\mathscr{K}-\frac{1}{2}m^2\partial_j\mathscr{K}-\frac{1}{2}m^2\partial_j'\mathscr{K}\right)\right|_\diag .
\end{multline*}
Now use the symmetric properties of $\mathscr{K}$ in equations \eqref{symmetry of A} and \eqref{A eqn 1}, we have
\begin{multline*}
\left.\left(\partial^i\partial_i'\partial_j\mathscr{K}+\partial'^i\partial_i'\partial_j\mathscr{K}-\frac{1}{2}\partial_j\partial_k'\partial^k\mathscr{K}-\frac{1}{2}\partial_j'\partial_k'\partial^k\mathscr{K}-\frac{1}{2}m^2\partial_j\mathscr{K}-\frac{1}{2}m^2\partial_j'\mathscr{K}\right)\right|_\diag\\
=\left.\left(\partial'^i\partial_i'\partial_j\mathscr{K}-\frac{1}{2}m^2\partial_j\mathscr{K}-\frac{1}{2}m^2\partial_j'\mathscr{K}\right)\right|_\diag .
\end{multline*}
The function $\mathscr{K}$ solves the Klein-Gordon equation, i.e. $$(\partial'^i\partial_i'-m^2)\mathscr{K}=(\partial^i\partial_i-m^2)\mathscr{K}=0 .$$ Altogether, we have
\begin{equation*}
\partial^i (T_\ren)_{ij}(t,x)=\frac{1}{2}\left.\left(\partial'^i\partial_i'\partial_j\mathscr{K}-m^2\partial_j\mathscr{K}+\partial^i\partial_i\partial_j'\mathscr{K}-m^2\partial_j'\mathscr{K}\right)\right|_\diag=0 ,
\end{equation*}
which shows the divergence-free property of the renormalised stress energy tensor.
\end{proof}
\section{The relative trace-formula and the Casimir energy}
\label{The relative trace-formula and the Casimir energy}
As mentioned in the introduction the renormalised stress-energy tensor $T_\ren(x)$ becomes unbounded and non-integrable when $x$ approaches the boundary of obstacles \cite{BrownMaclay,Candelas1982,DC1979,fulling2007vacuum}. This prevents us from defining a renormalised total energy. One way to circumvent the problem is to introduce the relative framework of \cite{RTF} which we now summarise. The main advantage of this construction is that it completely avoids ill-defined quantities and the need for regularisation. 

Relative quantities are defined with respect to different obstacle configurations where instead of $\calO$ only one obstacle $\calO_j$ is present, i.e. where $\calO$ is replaced by $\calO_j$. If an operator is defined with respect to such a configuration we use the subscript $\calO_j$, and we use the subscript $\calO$ to distinguish it from the original configuration.
For instance the renormalised stress energy tensor $T_\ren$ in Theorem \ref{expression of T} will be denoted by $(T_\ren)_\calO$, which shows its dependence on the presence of obstacles $\calO=\calO_1 \cup \dots \cup \calO_N$. Similarly, $(T_\ren)_{\calO_j}$ denotes the renormalised stress energy tensor when only obstacle $\calO_j$ is present and $\Delta_{\calO_j}$ denotes the Laplace operator with Dirichlet boundary condition on $\partial \calO_j$. The operator $H_{\calO_j}$ is defined in the same way.

Now we introduce a relative operator
\begin{equation}
\label{Hrel eqn}
H_\rel:=H_\calO-H_\free-\sum_{i=1}^N (H_{\calO_i}- H_\free).
\end{equation}
 More generally one defines the relative operator associated with a polynomially bounded function $f:[0,\infty) \to \R$, i.e.
\begin{align*}
D_{f,m}:=&f((-\Delta_\calO+m^2)^{\frac{1}{2}})-f((-\Delta_\free+m^2)^{\frac{1}{2}})\\&-\sum_{i=1}^N\left( f((-\Delta_{\calO_i}+m^2)^{\frac{1}{2}})-f((-\Delta_\free+m^2)^{\frac{1}{2}})\right).
\end{align*}
Since all our operators are densely defined operators on the same Hilbert space $L^2(\R^d)$ this combination makes sense. As a consequence of $f$ being polynomially bounded the space $C^\infty_0(X)$ is in the domain of $D_f$ and therefore $D_f$ has a distribution kernel in $\mathcal{D}'(X \times X)$.

To simplify our analysis later, we absorb the dependence of mass $m$ in the functional $f$. We could write $f_m$ to emphasise the dependence on $m$, but the later analysis will not be affected by $m$. Therefore, we omit the $m$-dependence and have   
\begin{equation*}
D_f:=f((-\Delta_\calO)^{\frac{1}{2}})-f((-\Delta_\free)^{\frac{1}{2}})-\sum_{i=1}^N\left( f((-\Delta_{\calO_i})^{\frac{1}{2}})-f((-\Delta_\free)^{\frac{1}{2}})\right).
\end{equation*}

The main result of \cite{RTF} is that for a large class of functions $f$, including the functions $f(\lambda)=\sqrt{\lambda^2}$ and $f(\lambda)=\sqrt{\lambda^2+m^2}$ which are of interest in our context, the operator $D_f$ 
is trace-class and its trace can be computed by integrating the kernel on the diagonal. We now explain the precise statement of this result and its relation to the determinant of the boundary layer operator.

In the following we will denote by $G_{\calO,\lambda} \in \mathcal{D}'(X \times X)$ the distributional kernel of the resolvent $R_{\calO,\lambda}=(-\Delta_\calO - \lambda^2)^{-1}$. The kernels $G_{\calO_j,\lambda}=\breve R_{\calO_j,\lambda}$ and $G_{\free,\lambda}=\breve R_{\free,\lambda}$ are defined in an analogous way.
By elliptic regularity these Green's distributions are smooth away from the diagonal $x=y$.

Recall that for $\lambda \in \C$ we have the single layer potential operator
$$
\calS_\lambda : C^{\infty}(\partial \calO)\to C^{\infty}(\overline{\calE})\oplus C^{\infty}(\overline{\calO})\subset C^{\infty}(\R^d\backslash \partial \calO)
$$
given by
$$
\calS_\lambda f(x)=\int_{\partial \calO} G_{\free,\lambda}(x,y) f(y) \dd \sigma(y),
$$
where $\dd \sigma$ is the surface measure.
Let $\gamma :H^{s+\frac{1}{2}}(\R^d)\to H^s(\partial \calO)$ be the Sobolev trace operator for $s>0$. 
Properties of the Sobolev trace operator can be found, for instance, in \cite{McLean2000}. One can write the above also as $\calS_\lambda=G_{\free,\lambda}\circ \gamma^*$. Restriction of $\calS_{\lambda}$ to the boundary defines an operator
$$
Q_\lambda f(x)=\int_{\partial \calO} G_{\free,\lambda}(x,y) f(y) \dd \sigma(y).
$$
The operator $Q_\lambda$ is known to have the following properties.

Since the boundary $\partial \calO$ consists of $N$ connected components $\partial \calO_j$, we therefore have an orthogonal decomposition $L^2(\partial \calO)=\oplus_{j=1}^N L^2(\partial \calO_j)$. Let $p_j: L^2(\partial \calO) \to L^2(\partial \calO_j)$ be the corresponding orthogonal projection. Now we have
\begin{equation}
\label{Q decomp}
\tilde{Q}_\lambda=\sum_{j=1}^N p_jQ_\lambda p_j, \quad
T_\lambda=\sum_{j\ne k}^N p_jQ_\lambda p_k \quad \text{and} \quad
Q_\lambda=\tilde{Q}_\lambda+T_\lambda .
\end{equation}
In short, $\tilde{Q}_\lambda$ and $T_\lambda$ are respectively the diagonal and off diagonal part of $Q_\lambda$. Now let $\mathfrak{D}_\nu$ to be a sector in the upper half plane and it is given by 
\begin{equation}
\label{Deps def}
\mathfrak{D}_\nu:=\{z\in \C\; | \; \nu \le \arg(z) \le \pi-\nu \},
\end{equation}
where it suffices to consider $0<\epsilon< \frac{\pi}{2}$ for our applications. 

The operator $Q_\lambda$ is invertible for $\Im(\lambda)>0$. Moreover  $Q_\lambda \tilde Q_\lambda^{-1} -1$ is trace-class and the Fredholm determinant $\det(Q_\lambda \tilde Q_\lambda^{-1})$ of $Q_\lambda \tilde Q_\lambda^{-1}$ can be used to define a function
$$
 \Xi(\lambda) = \log \det(Q_\lambda \tilde Q_\lambda^{-1})
$$
which is holomorphic in the upper half space and for some $\delta'>0$
 satisfies the bound
$$
 | \Xi(\lambda) | = O(e^{-\delta' \Im(z)}).
$$
See Theorem 1.6 of \cite{RTF} for the above bound in the sector of the form $\{ z \mid  \Im(z) \geq b |\Re(z)| \}$ for some $\delta>0$.

Assume $0<\theta < \frac{\pi}{2}$ and let $\mathfrak{S}_\theta$ be the open sector
$$
\mathfrak{S}_\theta=\{z \in \C \; | \; z \ne 0, |\arg (z)|< \theta \} .
$$
Let $\mathcal{P}_\theta$ be the set of functions that are holomorphic and polynomially bounded in $\mathfrak{S}_\theta$.
\begin{definition}
	\label{Pclass}
	The space $\tilde{\mathcal{P}}_\theta$ is defined to be the space of functions $f$ such that $f(\lambda)=g(\lambda^2)$ for some $g\in \mathcal{P}_{\theta+\epsilon}$ for some $\epsilon>0$ and there exists $a>0$ such that $f = O(|z|^a)$ if $|z|<1$
\end{definition}
\begin{figure}[h]
	\centering
	\includegraphics[scale=1.3]{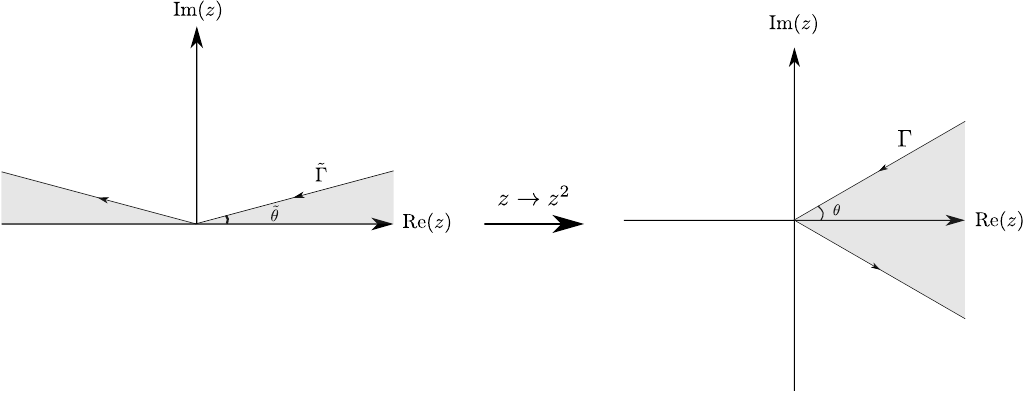}
	\caption{Definition of $\tilde{\Gamma}$.}
	\label{ztozsquare}
\end{figure}

We then have the following theorem.

\begin{theorem}[\cite{RTF}, Theorem 1.6 and 1.7]
	\label{frel trace thm}
	Let $f\in \tilde{\mathcal{P}}_\theta$. Then $D_f$ extends to a trace-class operator with integral kernel that is smooth on $X$ and has continuous inner and outer boundary values on $\partial \calO$.
	The trace of $D_f$ equals the integral over the diagonal of its integral kernel over $\R^d$. Moreover, it is equal to
	$$
	 \tr\,{D_f} = \frac{\rmi}{2\pi} \int_{\tilde \Gamma} f'(\lambda)\Xi(\lambda) \der \lambda.
	$$
\end{theorem}

In particular, choosing the function $f(\lambda) = \sqrt{\lambda^2 + m^2}-m$ one obtains $D_f = H_\rel$ and therefore $H_\rel$ is trace-class with trace equal to
$$
 \tr(H_\rel) = \frac{1}{ \pi} \int_{m}^\infty \frac{\omega}{ \sqrt{\omega^2 - m^2}} \Xi(\rmi \omega) \der \omega.
$$
This follows immediately by deforming the contour integral using the exponential decay of $\Xi$ in the upper half plane, considering the branch cut of $\sqrt{z+m^2}$ at $\rmi \sqrt{m}$.

\begin{definition}\label{relstress}
The relative stress energy tensor is the renormalised stress energy tensor in the relative setting and it is defined as
\begin{equation*}
T_\rel =(T_\ren)_\calO -\sum_{i=1}^N (T_\ren)_{\calO_i} ,
\end{equation*}
where $(T_\ren)_\calO$ is the renormalised stress energy tensor for obstacle $\calO$ and $(T_\ren)_{\calO_i}$ is the renormalised stress energy tensor for obstacle $\calO_i$, which are defined at the beginning of this section.
\end{definition}

\begin{rem}
One can also consider other versions of a relative stress energy tensor. For instance, one can define $\calO_A=\calO_1\cup\cdots \cup \calO_j$ and $\calO_B=\calO_{j+1}\cup\cdots \cup \calO_N$ for some $1\le j < N$, dropping the connectedness requirement of obstacles and work with
$$
T_{\rel,A,B}=(T_\ren)_{\calO}-(T_\ren)_{\calO_A}-(T_\ren)_{\calO_B}.
$$
The corresponding energy encodes the amount of work needed to separate the two obstacle configurations $A$ and $B$. This quantity can also be expressed in terms of $T_\rel$ for $A$ and $B$. It is easy to see that this equals
$$
T_{\rel,A,B} = (T_\rel)_{\calO}-(T_\rel)_{\calO_A}-(T_\rel)_{\calO_B},
$$
and therefore working with $T_\rel$ only does not result in a loss of generality.
\end{rem}
\begin{theorem}
	\label{integrability thm}
 $T_\rel$ is smooth on $X$ and extends smoothly to $\overline{\calE}$ as well as to $\overline{\calO}$. The function $(T_\rel)_{00}$ is integrable on $\R^d$.
\end{theorem}

\begin{proof}
By Theorem \ref{stress energy tensor} we have
$$
(T_\rel)_{00}=\frac{1}{2} \breve H_\rel|_\diag+\frac{1}{8}\Delta((\breve  H^{-1})_\rel|_\diag).
$$
The theorem was shown in \cite{RTF} for the part $\frac{1}{2} \breve H_\rel|_\diag$ and the same method of proof can also be applied to the second term. We provide the full details here for the sake of completeness.
We use two estimates proved in \cite{RTF} which we now explain. Recall that the relative resolvent is given by
 \begin{equation*}
R_{\rel,\lambda}:=R_{\calO,\lambda}-R_{\free,\lambda}-\sum_{j=1}^N (R_{\calO_j,\lambda}-R_{\free,\lambda}),
\end{equation*}
and $R_{\calO,\lambda}=(-\Delta_\calO-\lambda^2)^{-1}$, $R_{\free,\lambda}=(-\Delta_\free-\lambda^2)^{-1}$, and $R_{\calO_i,\lambda}=(-\Delta_{\calO_i}-\lambda^2)^{-1}$. For the integral kernel $\breve R_{\rel,\lambda}$ we write $G_{\rel,\lambda}$.

As shown in \cite{RTF} in the proof of Theorem 1.5 the integral kernel $G_{\rel,\lambda}$ of $R_{\rel, \lambda}$ is smooth up to the boundary on $\overline{\calE}$ as well as to $\overline{\calO}$ and its $C^k$-norms on compact subsets $K \subset \overline \calE \times \overline{\calE}$ satisfy the bound
$$
 \| G_{\rel, \lambda}  \|_{C^k(K)} \leq C_{K,k} (\log|\lambda|)^\ell e^{-\delta' \Im(\lambda)}
$$
for some $\ell\geq 0$ for $\lambda$ in the sector containing the imaginary axis.

We consider the operator

\begin{equation}
\label{Hrel and Rrel eqn}
(H^{-1})_\rel = H_{\calO}^{-1} - H^{-1}_\free - \sum_{j=1}^N (H_{\calO_j}^{-1} - H^{-1}_\free)
= \frac{2}{\pi}\int_m^\infty \frac{\lambda}{\sqrt{\lambda^2-m^2}} R_{\rel,\rmi \lambda} \der \lambda.
\end{equation}

A similar bound holds for $K \subset \overline{\calO} \times \overline{\calE}$, $K \subset \overline{\calE} \times \overline{\calO}$, and $K \subset \overline{\calO} \times \overline{\calO}$. The $\log$-factor in the estimate is only needed in dimension two.
It then follows that $(H^{-1})_\rel$ has an integral kernel that extends smoothly to $\overline{\calE}$ as well as to $\overline{\calO}$ (i.e. smooth up to the boundary). 

Similarly $\breve H_\rel$ is also smooth up to the boundary. Hence, by Theorem \ref{stress energy tensor} and Definition \ref{relstress}, we obtain the smoothness of $T_\rel$ up to the boundary. 
In order to show integrability we recall an estimate for the diagonal of the integral kernel of the resolvent difference, in particular
 \cite{RTF}*{Theorem 2.9, Equ. (21) and (22)}.

Let $\lambda\in \mathfrak{D}_\nu$ (See \eqref{Deps def}, i.e. a sector in the upper half plane) and $k_{\calO,\lambda}(x,y)$ denote the integral of $R_{\calO,\lambda}-R_{\free,\lambda}$, then we have
\begin{equation*}
k_{\calO,\lambda}(x,y)=-\langle G_{\free,\lambda}(x,\cdot),Q_\lambda^{-1}(\overline{G}_{\free,\lambda}(y,\cdot))\rangle,
\end{equation*}
which implies
\begin{align*}
|\nabla (k_{\calO,\lambda}(x,x))|
\le & C\sum_{j=1}^d \Big(| \left(\langle \partial_{j} G_{\free,\lambda}(x,\cdot),Q_\lambda^{-1}(\overline{G}_{\free,\lambda}(x,\cdot))\rangle\right)|
\\
&+| \left(\langle  G_{\free,\lambda}(x,\cdot),Q_\lambda^{-1}(\partial_{j}\overline{G}_{\free,\lambda}(x,\cdot))\rangle\right)|\Big)
\\
\le  & C | \left(\langle \nabla G_{\free,\lambda}(x,\cdot),Q_\lambda^{-1}(\overline{G}_{\free,\lambda}(x,\cdot))\rangle\right)|,
\end{align*}
and
\begin{align*}
|\Delta (k_{\calO,\lambda}(x,x))|
\le C &\sum_{j=1}^d | \partial_{j}\partial_{j}\left(\langle  G_{\free,\lambda}(x,\cdot),Q_\lambda^{-1}(\overline{G}_{\free,\lambda}(x,\cdot))\rangle\right)|
\\
\le C &\Bigg( | \langle  \nabla G_{\free,\lambda}(x,\cdot),Q_\lambda^{-1}(\nabla\overline{G}_{\free,\lambda}(x,\cdot))\rangle|
\\
&+|\langle \Delta G_{\free,\lambda}(x,\cdot),Q_\lambda^{-1}(\overline{G}_{\free,\lambda}(x,\cdot))\rangle|\Bigg)
\\
\le C &\Bigg( \|\nabla G_{\free,\lambda}(x,\cdot)\|_{H^{\frac{1}{2}}}\|Q_\lambda^{-1}\|_{H^{\frac{1}{2}}\to H^{-\frac{1}{2}}}\|\nabla\overline{G}_{\free,\lambda}(x,\cdot)\|_{H^{\frac{1}{2}}}
\\
&+\|\Delta G_{\free,\lambda}(x,\cdot)\|_{H^{\frac{1}{2}}}\|Q_\lambda^{-1}\|_{H^{\frac{1}{2}}\to H^{-\frac{1}{2}}}\|\overline{G}_{\free,\lambda}(x,\cdot))\|_{H^{\frac{1}{2}}}\Bigg).
\end{align*}
Let $\rho(x)=\dist(x,\partial \calO)$. By Lemma A.2 of \cite{RTF} we have for $\rho(x)\lambda\le 1$,
\begin{equation*}
\sup_{y\in\partial \calO}|\partial^\alpha_x \partial^\beta_y G_{\free,\lambda}(x,y)|\le C
\left\{
\begin{aligned}
&\rho^{-(d-2+|\alpha|+|\beta|)}(x)&(d\ge 3),\\
&|\ln(\rho(x)\lambda)|+\rho^{-(|\alpha|+|\beta|)}(x)&(d=2),
\end{aligned}
\right.
\end{equation*} 
and $\ln(\rho(x)\lambda)$ disappears when $|\alpha|+|\beta|\ge 1$. Moreover, for $\rho(x)\lambda> 1$, one has
\begin{equation*}
\sup_{y\in\partial \calO}|\partial^\alpha_x \partial^\beta_y G_{\free,\lambda}(x,y)|\le C |\lambda|^{d-2+|\alpha|+|\beta|}e^{\Im \lambda \rho(x)}.
\end{equation*} 
This shows
\begin{equation*}
d\ge 3, \rho(x)>1\implies \left\{
\begin{aligned}
&\| G_{\free,\lambda}(x,\cdot)\|_{H^{\frac{1}{2}}}\le C_1 e^{-C_2(\Im \lambda)\rho(x)}\rho^{-(d-2)}(x),\\
&\|\nabla G_{\free,\lambda}(x,\cdot)\|_{H^{\frac{1}{2}}}\le C_1 e^{-C_2(\Im \lambda)\rho(x)}\rho^{-(d-1)}(x),\\
&\|\Delta G_{\free,\lambda}(x,\cdot)\|_{H^{\frac{1}{2}}}\le C_1 e^{-C_2(\Im \lambda)\rho(x)}\rho^{-d}(x),
\end{aligned}
\right.
\end{equation*}
and
\begin{equation*}
d=2, \rho(x)>1\implies \left\{
\begin{aligned}
&\| G_{\free,\lambda}(x,\cdot)\|_{H^{\frac{1}{2}}}\le C_1 e^{-C_2(\Im \lambda)\rho(x)}|\ln(\rho(x)\lambda)|,\\
&\|\nabla G_{\free,\lambda}(x,\cdot)\|_{H^{\frac{1}{2}}}\le C_1 e^{-C_2(\Im \lambda)\rho(x)}\rho^{-1}(x),\\
&\|\Delta G_{\free,\lambda}(x,\cdot)\|_{H^{\frac{1}{2}}}\le C_1 e^{-C_2(\Im \lambda)\rho(x)}\rho^{-2}(x),
\end{aligned}
\right.
\end{equation*}
for some positive $C_1$ and $C_2$. By Corollary 2.8 of \cite{RTF} we have
$$\|Q_\lambda^{-1}\|_{H^{\frac{1}{2}}\to H^{-\frac{1}{2}}}\le C(1+|\lambda|^2).$$
Now we can conclude that 
\begin{equation}
\label{d3kernel}
d\ge 3, \rho(x)>1\implies \left\{
\begin{aligned}
&| k_{\calO,\lambda}(x,x)|\le C_1 \rho^{-2d+4}(x) e^{-C_2 \Im\lambda \rho(x)},\\
&|\nabla (k_{\calO,\lambda}(x,x))|\le C_1 \rho^{-2d+3}(x) e^{-C_2 \Im\lambda \rho(x)},\\
&|\Delta (k_{\calO,\lambda}(x,x))|\le C_1 \rho^{-2(d-1)}(x) e^{-C_2 \Im\lambda \rho(x)},
\end{aligned}
\right.
\end{equation}
and
\begin{equation}
\label{d2kernel}
d=2, \rho(x)>1\implies \left\{
\begin{aligned}
&| k_{\calO,\lambda}(x,x)|\le C_1 e^{-C_2(\Im \lambda)\rho(x)}|\ln(\rho(x)\lambda)|^2,\\
&|\nabla (k_{\calO,\lambda}(x,x))|\le C_1 e^{-C_2(\Im \lambda)\rho(x)}|\ln(\rho(x)\lambda)|\rho^{-1}(x),\\
&|\Delta (k_{\calO,\lambda}(x,x))|\le C_1 e^{-C_2 \Im\lambda \rho(x)}|\ln(\rho(x)\lambda)|\rho^{-2}(x).
\end{aligned}
\right.
\end{equation}
Let $h_\rel=(\breve H^{-1})_\rel|_\diag$. That is
$h_\rel(x)=\frac{2}{\pi}\int_m^\infty \frac{\lambda}{\sqrt{\lambda^2-m^2}} G_{\rel,\rmi \lambda}(x,x)\der \lambda$.
For $m\neq 0$, one can use equations \eqref{d3kernel} and \eqref{d2kernel} to get the decay rate of $h_\rel(x)$ by integrating over $\lambda$. That is, for $\rho(x)>1$, $h_\rel (x)$ has a decay of $\rho^{-k}(x)e^{-Cm\rho(x)}$ with both $k$ and $C$ being positive. This warrants the integrability of $\Delta h_\rel (x)$ for $d\ge 2$, $m\neq 0$ and $\rho(x)>1$. Therefore, we will now focus on the case $m=0$. By integrability of the integrand we can interchange differentiation and integration and therefore get
\begin{align*}
&\nabla h_\rel (x)=\nabla \left(\frac{2}{\pi}\int_0^\infty  G_{\rel,\rmi \lambda}(x,x) \der \lambda\right)
=\frac{2}{\pi}\int_0^\infty  \nabla\left(G_{\rel,\rmi \lambda}(x,x)\right) \der \lambda,\\
&\Delta h_\rel (x)=\Delta \left(\frac{2}{\pi}\int_0^\infty  G_{\rel,\rmi \lambda}(x,x) \der \lambda\right)
=\frac{2}{\pi}\int_0^\infty  \Delta\left(G_{\rel,\rmi \lambda}(x,x)\right) \der \lambda.
\end{align*}
Again, by integrating over $\lambda$, we have
\begin{equation}
\label{DofHrel}
d\ge 2, \rho(x)>1\implies \left\{
\begin{aligned}
&|\nabla h_\rel (x)|\le C \rho^{-2d+2}(x),\\
&|\Delta h_\rel (x)|\le C \rho^{-2d+1}(x).
\end{aligned}
\right.
\end{equation}

Let $\Omega\subset \R^d$ be an open set with $\dist(\Omega,\partial \calO)>1$ and $\varphi\in C_0^\infty(\R^d)$ satisfy $0\le \varphi \le 1$ and $\varphi=1$ in a neighbourhood of $\R^d\backslash\Omega$. Then we have the decomposition
\begin{equation}
\label{hrel decomp eqn}
\Delta h_\rel=\Delta[(1-\varphi)h_\rel]+\Delta (\varphi h_\rel).
\end{equation}
The integrability of $\Delta (\varphi h_\rel)$ in equation \eqref{hrel decomp eqn} follows from the smoothness property of the kernel of $(\breve H^{-1})_\rel$ at the diagonal as shown above. Thereby the integrability of $\Delta h_\rel$ on $\R^d$ is equivalent to the one of $\Delta[(1-\varphi)h_\rel]$ on $\supp(1-\varphi)$. This follows immediately from equation \eqref{DofHrel}. Therefore, we have shown the integrability of $\Delta((\breve H^{-1})_\rel|_\diag)$ on $\R^d$. Finally, the integrability of $T_\rel$ follows from Theorem \ref{expression of T} and the definition of relative stress energy tensor (i.e. Definition \ref{relstress}).

\end{proof}

\begin{definition}
The relative energy is defined as
\begin{equation*}
E_{\rel}=\int_{\R^d} (T_\rel)_{00} \dd x=\int_{\R^d} \left( ((T_\ren)_\calO)_{00} -\sum_{i=1}^N ((T_\ren)_{\calO_i})_{00}\right) \dd x .
\end{equation*}
\end{definition}

\begin{theorem}
\label{Erel thm}
We have the equality
\begin{equation}
\label{Erel eqn}
E_\rel = \frac{1}{2} \Tr\, H_{\rel} = \frac{1}{2 \pi} \int_{m}^\infty \frac{\omega}{ \sqrt{\omega^2 - m^2}} \Xi(\rmi \omega) \der \omega.
\end{equation}
\end{theorem}
\begin{proof}
We have
$$
\int_{\R^d} \Delta((\breve H^{-1})_\rel|_\diag) \dd x=\lim_{r \to \infty}\int_{B_r} \Delta((\breve H^{-1})_\rel|_\diag) \dd x,
$$
where $B_r$ is the ball of radius $r$ centred at the origin.
As $\Delta((\breve H^{-1})_\rel|_\diag)$ has only jump-type discontinuity across $\partial{X}$, we can apply the divergence theorem to the integral on $B_r$ for sufficiently large $r$. That is
\begin{multline*}
\int_{B_r} \Delta((\breve H^{-1})_\rel|_\diag) \dd x=\int_{\partial B_r}\partial_\nu((\breve H^{-1})_\rel|_\diag) \dd \sigma(x)\\+\int_{\partial X} \left(\partial_\nu((\breve H^{-1})_\rel|_\diag)\right)_+ \dd \sigma(x)+\int_{\partial X}\left(\partial_\nu((\breve H^{-1})_\rel|_\diag)\right)_- \dd \sigma(x),
\end{multline*}
where $(\cdot)_+$ and $(\cdot)_-$ are the exterior and interior limits respectively. 
From equation \eqref{DofHrel}, we also have
$$
|\nabla[(1-\varphi)h_\rel]|(x)\le \frac{C_\Omega}{(\dist(x,\partial \calO))^{2d-2}},
$$
which implies the contribution of the integral over $\partial B_r$ vanished as $r \to \infty$ and therefore
$$
\int_{\R^d} \Delta((\breve H^{-1})_\rel|_\diag) \dd x=\int_{\partial X} \left(\partial_\nu((\breve H^{-1})_\rel|_\diag)\right)_+ \dd \sigma(x)+\int_{\partial X}\left(\partial_\nu((\breve H^{-1})_\rel|_\diag)\right)_- \dd \sigma(x).
$$

From equation \eqref{Hrel and Rrel eqn}, we then have
\begin{align*}
\int_{\R^d} \nabla \cdot \nabla ((\breve H^{-1})_\rel|_\diag) \dd x
&=\int_{\R^d} \Delta((\breve H^{-1})_\rel|_\diag) \dd x
\\
&=\frac{2}{\pi}\int_{\R^d}\int_0^\infty  \left[\nabla \cdot \nabla (G_{\rel,\rmi \lambda}|_\diag)\right] \dd \lambda \dd x
\\
&=\frac{2}{\pi}\int_{\R^d}\int_0^\infty  \left[2\nabla \cdot (\nabla G_{\rel,\rmi \lambda}|_\diag)\right] \dd \lambda \dd x
\\
&=2\int_{\R^d}\left[\nabla \cdot (\nabla (\breve H^{-1})_\rel|_\diag)\right] \dd x .
\end{align*}

This shows
$$
\frac{1}{2}\int_{\R^d} \Delta((\breve H^{-1})_\rel|_\diag) \dd x=\int_{\partial X} \left(\partial_\nu (\breve H^{-1})_\rel|_\diag\right)_+ \dd \sigma(x)+\int_{\partial X}\left(\partial_\nu (\breve H^{-1})_\rel|_\diag\right)_- \dd \sigma(x).
$$
We start by showing that  for the restrictions to $\partial \calO_q$ we have the following identity
$$
\left[\left.\partial_\nu\left(\breve H_{\calO}^{-1} - \breve H^{-1}_\free -(\breve H_{\calO_q}^{-1} - \breve H^{-1}_\free)\right)\right|_\diag\right]_{\partial \calO_q,\pm}= \left[\left.\partial_\nu\left(\breve H_{\calO}^{-1}  -\breve H_{\calO_q}^{-1} \right)\right|_\diag\right]_{\partial \calO_q,\pm} =0.
$$
To see this we temporarily denote by $k(x,y)$ the integral kernel of  $\left(H_{\calO}^{-1}  - H_{\calO_q}^{-1} \right)$. This kernel vanishes $\calO_q \times \calO_q$ and the interior normal derivative therefore vanishes trivially. We therefore only need to concern ourself with the exterior normal derivative. As shown in the proof of Theorem \ref{integrability thm} the kernel $(H^{-1})_\rel$ is smooth. One concludes from this, using Theorem \ref{smooth K thm}, that $\left(H_{\calO}^{-1}  - H_{\calO_q}^{-1} \right)$ is smooth near $\calO_q \times \calO_q$ in $\overline \calE \times \overline \calE$.
The kernel $k$ satisfies Dirichlet boundary conditions in both variables in the sense that $k(x,y)=0$ if $y \in \partial \calO_q$ or if $x \in \partial \calO_q$. By the chain rule   $\left[\left.\partial_\nu\left(\breve H_{\calO}^{-1}  - \breve H_{\calO_q}^{-1} \right)\right|_\diag\right]_\pm(x)$
equals  $(\partial_{\nu,x} k(x,y) + \partial_{\nu,y} k(x,y))|_{y=x}$, which therefore vanishes on  $\partial \calO_q$. We then have

\begin{align*}
&\int_{\partial X} \left(\partial_\nu (\breve H^{-1})_\rel|_\diag \right)_\pm \dd \sigma(x)
\\
=&\sum_{q=1}^N\int_{\partial \calO_q} \left(\partial_\nu (\breve H^{-1})_\rel|_\diag \right)_\pm \dd \sigma(x)
\\
=&\sum_{q=1}^N\int_{\partial \calO_q}\left[\left.\partial_\nu\left(\breve H_{\calO}^{-1} - \breve  H^{-1}_\free - \sum_{p=1}^N (\breve H_{\calO_p}^{-1} - \breve H^{-1}_\free)\right)\right|_\diag\right]_\pm \dd \sigma(x)
\\
=&\sum_{q=1}^N \int_{\partial \calO_q}\left[-\sum_{p\ne q}^N  \left.\partial_\nu\left(\breve H_{\calO_p}^{-1}- \breve H^{-1}_\free \right)\right|_\diag\right]_\pm \dd \sigma(x).
\end{align*}

By Theorem \ref{smooth K thm} we know that $\left.\partial_\nu\left(\breve H_{\calO_p}^{-1}- \breve H^{-1}_\free \right)\right|_\diag$ is smooth across the boundary $\partial\calO_q$ for $p \ne q$, which implies
$$
 \int_{\partial \calO_q}\left[-\sum_{p\ne q}^N  \left.\partial_\nu\left(\breve H_{\calO_p}^{-1}- \breve H^{-1}_\free \right)\right|_\diag\right]_+ \dd \sigma(x)=- \int_{\partial \calO_q}\left[-\sum_{p\ne q}^N  \left.\partial_\nu\left(\breve H_{\calO_p}^{-1}- \breve H^{-1}_\free \right)\right|_\diag\right]_- \dd \sigma(x).
$$
Hence we have $\int_{\R^d} \Delta((\breve H^{-1})_\rel|_\diag) \dd x=0$, which verifies the first equation in \eqref{Erel eqn}. The representation of the trace of $H_{\rel}$ in terms of $\Xi$ follows via $ \Xi'(\lambda) = -2\lambda \tr\left( R_{\rel,\lambda} \right)$ from Theorem 4.2 in \cite{RTF}, Theorem \ref{frel trace thm} and equation \eqref{Hrel and Rrel eqn}. 

\end{proof}

\section{Estimates on the relative resolvent}
\label{Estimates on the relative resolvent}
In preparation for the proof of the variational formula we will need some additional estimates on the relative resolvent
 $R_{\rel,\lambda}$, which we collect in this section. We have the following well known layer potential representation (see for example \cite{RTF}*{Equ. (19)})
\begin{equation}
\label{Rdiff}
R_{\calO, \lambda} - R_{\free, \lambda}=-\calS_\lambda Q_\lambda^{-1}\calS_\lambda^t,
\end{equation}
which gives
\begin{equation}
\label{Rrel}
R_{\rel,\lambda}=-\calS_\lambda (Q_\lambda^{-1}-\tilde{Q}_\lambda^{-1}) \calS_\lambda^t
\end{equation}
Also, let $\rho$ be the function defined as 
\begin{equation*}
\rho(\lambda)
:=\left\{
\begin{aligned}
&\lambda^{d-4} & &\text{if $0< \lambda \le 1$ and $d=2,3$,}\\
&|\log \lambda| +1 & &\text{if $0< \lambda \le 1$ and $d=4$,}\\
&1 & &\text{if $0< \lambda \le 1$ and $d>4$,}\\
&1 & &\text{if $\lambda>1$}.
\end{aligned}
\right.
\end{equation*}

For $i,j,k\in \mathbb{Z}$, let $\rho_{i,j;k}$ be the functions defined as
\begin{equation*}
\rho_{i,j;k}(\lambda):=
\left\{
\begin{aligned}
&(|\log \lambda|+1)^i \lambda^j& &\text{for $0<\lambda \le 1$},\\
&\lambda^{k} & &\text{for $\lambda>1$}.
\end{aligned}
\right.
\end{equation*}
The following proposition partially follows from \cite[Proposition 2.1]{RTF} and extends \cite[Proposition 2.2]{RTF}.
Let us summarise some mapping properties of the layer potential operators.

\begin{proposition} \label{S properties}
	Let $d\ge 2$ and $-\frac{3}{2}\le s \le -\frac{1}{2}$. 
	Then we have the following properties of $\calS_{\lambda}$ for $\lambda \in \mathfrak{D}_\nu$ (See \eqref{Deps def}, i.e. a sector in the upper half plane).
	\
	
	\begin{enumerate}
		\item \label{S prop 1} If $\Re(\lambda) \ne 0$ then $\|\calS_{\lambda}\|_{H^{s}(\partial \calO)\to L^2(\R^d)}\le C_s \frac{\sqrt{1+|\lambda|^2}}{\Re(\lambda)\Im(\lambda)}$. 
		\item \label{S prop 2} Let $\chi\in C^\infty(\R^d)$ be supported in $V$, where $V\subset \R^d$ is an open set with smooth boundary and $\dist(V,\partial \calO)>0$. For $0<\delta < \dist(V,\partial \calO)$, we have that $\chi \calS_{\lambda}: H^{s}(\partial \calO)\to L^2(V)$ is a Hilbert-Schmidt operator whose Hilbert-Schmidt norm is bounded by $$\|\chi \calS_{\lambda}\|_{HS(H^{s}(\partial \calO)\to L^2(\R^d))} \le C_{\delta,\nu,k}\, \sqrt{\rho(\Im \lambda)} e^{-\delta \Im \lambda}$$ which implies $$\|\calS_{\lambda}\|_{H^{s}(\partial \calO)\to L^2(\R^d)} \le C_{\delta,\nu,s}\left(1+\sqrt{\rho(\Im \lambda)}\right).$$
		\item \label{S prop 3} $\|Q_\lambda^{-1}\|_{H^{\frac{1}{2}}(\partial \calO) \to H^{-\frac{1}{2}}(\partial \calO)} \le C_\nu (1+(\Im \lambda)^2)^2.$
	\end{enumerate}
\end{proposition}
\begin{proof}
	Recall that we could also write the single layer potential operator as $\calS_\lambda=G_{\free,\lambda}\circ \gamma^*$. Note that $\gamma^*: H^{s}(\partial \calO) \to H^{s-\frac{1}{2}}(\R^d)$, $G_{\free,\lambda}: H^{s-\frac{1}{2}}(\R^d) \to H^{s+\frac{3}{2}}(\R^d)$ and the natural inclusion map $\iota: H^{s+\frac{3}{2}}(\R^d) \to L^2(\R^d)$ are bounded maps. For $-\frac{3}{2}\le s \le -\frac{1}{2}$, their norms are bounded by
	\begin{equation*}
	\left\{
	\begin{aligned}
	&\|\gamma^*\|_{H^s \to H^{s-\frac{1}{2}}}\le C_s , \\
	&\|G_{\free,\lambda}\|_{H^{s-\frac{1}{2}} \to H^{s+\frac{3}{2}}}\le \sup_{x\in\R}\left|\frac{\sqrt{1+x^2}}{x^2-\lambda^2}\right|\le \frac{\sqrt{1+|\lambda|^2}}{\Re(\lambda)\Im(\lambda)} ,\\
	&\|\iota\|_{H^{s+\frac{3}{2}}\to H^0}\le 1 .
	\end{aligned}
	\right.
	\end{equation*}
	The second bound follows from the spectral representation of $-\Delta_\free$. Therefore, we conclude the proof for part \eqref{S prop 1}. 
	
	Part \eqref{S prop 3} follows immediately from the bound on the Dirichlet-to-Neumann operator in \cite{RTF}. For part \eqref{S prop 2}, let $\calO\subset V'\subset \R^d$ be an open set with $\dist(V,V')>0$ and choose $0<\delta<\dist(V,V')$. In particular, $0<\delta<\dist(V,\partial\calO)$.  Let $k\in \N$. Now we have from \cite{RTF} the estimate
	\begin{multline*}
	\|(-\Delta)^{\frac{k}{2}} G_{\free,\lambda}\|^2_{L^2(V\times V')}
	\\
	\le C_{\delta,\nu,k}
	\left\{
	\begin{aligned}
	&(\Im \lambda)^{2k-2}\left(1+\left(\Im \lambda \log \left(\Im\lambda\right)\right)^2\right)e^{-2\delta \Im \lambda} & & \text{for $d=2$}\\
	&(\Im \lambda)^{2k}\left(1+\log\left(\Im \lambda\right)\right)e^{-2\delta \Im \lambda} & & \text{for $d=4$}\\
	&(\Im \lambda)^{d+2k-4}\left(1+|\Im \lambda |^{4-d}\right)e^{-2\delta \Im \lambda} & & \text{for $d\ne 2,4$}\\
	\end{aligned}
	\right.
	\end{multline*}
	where $\Delta=\Delta_x+\Delta_y$ is the Laplace operator on $V\times V'$. In other words, we have
	\begin{equation*}
	\|G_{\free,\lambda}\|^2_{H^{k}(V\times V')}\le C_{\delta,\nu,k}\,(1+(\Im \lambda)^{2k})\rho(\Im \lambda)e^{-2\delta \Im \lambda}.
	\end{equation*}
	Now by taking Sobolev trace, we have
	\begin{equation*}
	\|G_{\free,\lambda}\|^2_{H^{s}(V\times \partial \calO)}\le C_{\delta,\nu,k}\,(1+(\Im \lambda)^{2s+1})\rho(\Im \lambda)e^{-2\delta \Im \lambda}.
	\end{equation*}
	Statement  \eqref{S prop 2} follows, using the properties of the Hilbert-Schmidt norm (for example Section A.3.1 in \cite{ShubinPseudos}). Since $-\frac{3}{2}\le s \le -\frac{1}{2}$, we have
	\begin{equation}
	\label{S prop 2 eqn 1}
	\begin{aligned}	
	\|\chi \calS_{\lambda}\|^2_{HS(H^{s}(\partial \calO)\to L^2(V))} &\le C_1 \|\chi \calS_{\lambda}E_{-s}\|^2_{HS(L^2(\partial \calO) \to L^2(V))}
	\\
	&\le C_2 \| G_{\free,\lambda} \|^2_{H^{-s}(V\times \partial \calO)}
	\\
	&\le C_{\delta,\nu,k}\,(1+(\Im \lambda)^{-2s+1})\rho(\Im \lambda)e^{-2\delta \Im \lambda}
	\\
	&\le C_{\delta,\nu,k}\, \rho(\Im \lambda) e^{-2\delta \Im \lambda},
	\end{aligned}
	\end{equation}
	where $E_s=(\sqrt{-\Delta_{\partial\calO}}+1)^s$. To prove the last property of $\calS_{\lambda}$, we are left with proving the bound for $\phi\, \calS_{\lambda}$, where $\phi \in C^\infty(V'')$ is supported in $V''$ with $\calO\subset V'\subset V''$, $\phi(x)=1$ on $\calO$ and $1-\phi$ is supported in $V$. It suffices to bound $\phi \, G_{\free,\lambda} \, \phi : H^{s}(\R^d) \to L^2(\R^d)$ for $-\frac{3}{2}\le s \le -\frac{1}{2}$. Note that the explicit kernel of $\phi G_{\free,\lambda}\phi$, denoted by $K_\lambda(x,y)$, is given by
	\begin{equation*}
	K_\lambda(x,y)=\frac{i}{4}\phi(x)\phi(y)\left(\frac{\lambda}{2\pi |x-y|}\right)^{\frac{d-2}{2}}H^{(1)}_{\frac{d-2}{2}}(\lambda|x-y|),
	\end{equation*}
	where $H^{(1)}_\nu$ is the Hankel function. By the Schur test and estimates on the free resolvent (see Appendix \ref{Bounds on the free resolvent} for details), we have
	\begin{equation*}
	\|\phi G_{\free,\lambda}\phi\|_{L^2\to H^2}\le C_{r_0,\phi}
	\left\{
	\begin{aligned}
	&\rho_{1,0;0}(\Im \lambda) & &\text{for $d=2$},\\
	&1 & &\text{for $d\ne 2$}.
	\end{aligned}
	\right.
	\end{equation*}
	Taking the adjoint, we get the same bound for $H^{-2}\to L^2$. Finally, using the estimate \eqref{S prop 2 eqn 1} with $\chi=1-\phi$, we obtain
	$$
	\|\calS_{\lambda}\|_{H^{s}(\partial \calO)\to L^2(\R^d)} \le C_{\delta,\nu,s}\left(1+\sqrt{\rho(\Im \lambda)}\right)\quad \text{for $-\frac{3}{2}\le s \le -\frac{1}{2}$}.
	$$
\end{proof}
\begin{rem}
In the case of odd dimensions an easier argument can be used to provide a weaker estimate that is still sufficient for the purposes of this paper. Using the strong Huygens principle  one deduces (\cite[Section 3.1]{DZ2019})
\begin{equation}
\label{phiGphiodd}
\|\phi \, G_{\free,\lambda} \, \phi\|_{H^{-2} \to L^2}\le C_{\phi}(1+|\lambda|) \quad\text{for}\quad \lambda\in \mathfrak{D}_\nu ,
\end{equation}
which implies
$$
\|\phi \, G_{\free,\lambda} \, \phi\|_{H^{s-\frac{1}{2}} \to L^2}\le C_{\phi}(1+|\lambda|) \le C_{\phi,\nu}(1+\Im \lambda) \quad\text{for}\quad -\frac{3}{2}\le s \le -\frac{1}{2}.
$$
This gives $\|\phi \, \calS_{\lambda}\|_{H^{s} \to L^2}\le C_{\phi,\nu}(1+\Im \lambda)$ for $-\frac{3}{2}\le s \le -\frac{1}{2}$.
Applying inequality \eqref{S prop 2 eqn 1} for $\chi=1-\phi$ and using $\frac{-2s+1}{2} \ge 1$, we get
$$
\|\calS_{\lambda}\|_{H^{s}\to L^2} \le C_{\delta,\nu,s}\left(1+\Im \lambda+\sqrt{\rho(\Im \lambda)}\right)\quad \text{for $-\frac{3}{2}\le s \le -\frac{1}{2}$}.
$$
\end{rem}

\section{Hadamard variation formula and equivalence of approaches $(2)$ and $(3)$}
\label{Had:Sec}
In this section, we will show that a version of Hadamard variation of the renormalised stress energy tensor $(T_\ren)_{ij}$ defined in Section \ref{The renormalised stress-energy tensor} is related to the Hadamard variation formula for the resolvent. We will follow the methods developed in \cite{GarabedianSchiffer,Henry2005,HezariZelditch,Peetre} and derive a Hadamard variation formula for the relative resolvent, then apply it to the relative stress energy tensor. Short proofs of Hadamard variation formula can be found in \cite{Ozawa1982,SuzukiTsuchiya2016} for the case of bounded domains. Since we are dealing with an unbounded domain we extend theory to non-compact setting for the special case of boundary translation flows (see Definition \ref{def of BT flow}) in Theorem \ref{existence thm}, which are sufficient for our purposes.

\subsection{Hadamard variation formula}
Let $U$ be a possibly unbounded open subset in $\R^d$ with smooth compact boundary and $Y$ be a smooth vector field on $\R^d$. The flow, denoted dy $\varphi_\epsilon$ and generated by $Y$, gives,  for small $|\epsilon|$, a one-parameter family of smooth manifolds in $\R^d$, which is denoted by $U_\epsilon=\varphi_\epsilon(U)$. For our application, $U$ would be either $\mathcal{E}$ or $\calO$ as defined in Section \ref{Introduction} and we will only consider flows that generate rigid translations of obstacles. 
\begin{definition}
	\label{def of BT flow}
A flow $\varphi_\epsilon$ associated with vector field $Y$ is called a boundary translation if $Y$ is locally constant near $\partial U$.
\end{definition}
Following Peetre's derivation of Hadamard variation formula, we define the following variational derivative.
\begin{definition}
Let $u_\epsilon$ be a (weak-$*$) $C^1$ curve of functions in $\mathcal{D}'(U_\epsilon)$. The variational derivative at $\epsilon=0$ is given by
\begin{equation}
\label{variation derivative}
\delta_Y u:=\theta_Y u -Yu,
\end{equation}
where 
$
\theta_Y u=\lim_{\epsilon\to 0}\frac{\varphi_\epsilon^*u_\epsilon-u_0}{\epsilon}
$
is the variational derivative defined by Garabedian-Schiffer's in \cite{GarabedianSchiffer}. Here the derivative $\theta_Y u$ is understood in the weak-$*$-sense and the action of the vector-field $Y$ is understood in the sense of distribution.
\end{definition}
Note that $\theta_Y u$ is different from the standard (conventional) Lie derivative, as $u_\epsilon$ may have an additional dependence on $\epsilon$. In fact, the last term, $Yu$ in equation \eqref{variation derivative}, should be understood as the conventional Lie derivative of $u_0$.

The derivation of Hadamard variational formula for the resolvent associated with the Dirichlet Laplace operator usually starts with the energy quadratic form (see \cite{GarabedianSchiffer,HezariZelditch,Peetre}). The energy quadratic form associated with the Dirichlet Laplacian $-\Delta_{U_\epsilon}$ on $U_\epsilon$ is given by 
\begin{equation}
\label{energy form 1}
E_\epsilon(u,v;\lambda)=\int_{U_\epsilon} (\nabla u \cdot \overline{\nabla  v} -\lambda^2 u \overline{v} )\dd x ,
\end{equation}
where $\lambda\in\mathbb{C}$ and $\Im{\lambda}>0$. Using the diffeomorphism flow $\varphi_\epsilon$, one can pull-back the quadratic form from $U_\epsilon$ to $U_0=U$, which gives a one-parameter family of quadratic forms on $U$, i.e.
\begin{equation}
\label{energy form 2}
\tilde{E}_\epsilon(u,v;\lambda)=\int_{U} \left[\varphi_\epsilon^*\left(\nabla \varphi_\epsilon^{-1*}(u)\cdot \overline{\nabla\varphi_\epsilon^{-1*}(v)}\right) -\lambda^2 u \overline{v} \right] \varphi_\epsilon^*(\dd x) .
\end{equation}
Note that the energy form \eqref{energy form 1} and the induced energy forms \eqref{energy form 2} are related by
\begin{equation}
\label{energy form 3}
E_\epsilon(u,v;\lambda)=\tilde{E}_\epsilon(\varphi_\epsilon^*(u),\varphi_\epsilon^*(v);\lambda).
\end{equation}
The operator associated with the energy form \eqref{energy form 1} is the Dirichlet Laplace operator, whereas $\tilde{E}_\epsilon$ defines a one-parameter families of elliptic operators on $U$ for sufficiently small $\epsilon$. Let $G_{\lambda,\epsilon} =  G_{U_\epsilon,\lambda}$ be kernel of the resolvent for the Dirichlet Laplacian on $U_\epsilon$. Then from equations \eqref{energy form 1}, \eqref{energy form 2} and \eqref{energy form 3}, we have in the sense of distributions
\begin{equation}
\label{resolvent equation}
\left\{
\begin{aligned}
&(-\Delta_{x}-\lambda^2)G_{\lambda,\epsilon}(x,y)=\delta_{y}(x) \quad\text{in}\quad U_\epsilon\\
&G_{\lambda,\epsilon}(x,y)=0 \quad\text{for}\quad x \in\partial U_\epsilon
\end{aligned}
\right.,
\end{equation}
where $y$ is in the interior of $U_\epsilon$ and by elliptic regularity $G_{\lambda,\epsilon}(x,y)$ is then smooth at the boundary and therefore it makes sense to define its boundary value. 

As we would like to study the variation of resolvents, it is convenient to consider the variation as distributions on $U \times U$. In other words, for $R_\epsilon \in \calD'(U_\epsilon \times U_\epsilon)$ and $u, v\in C_0^\infty(U)$, we have, from the Schwartz kernel theorem,
\begin{equation}
\label{Var of R}
 \varphi_{\epsilon}^{*}R_ \epsilon (u \otimes v) = R_ \epsilon (\varphi_{-\epsilon}^{*}u \otimes \varphi_{-\epsilon}^{*} v ) = \langle \varphi_{-\epsilon}^{*}u , R_ \epsilon \varphi_{-\epsilon}^{*} \overline{v} \rangle=\langle u , \varphi_{\epsilon}^{*} R_ \epsilon \varphi_{-\epsilon}^{*} \overline{v} \rangle ,
\end{equation}
where the first two brackets correspond to the pairing between distributions and test functions while the third and the forth brackets are the $L^2$ inner products on $U_\epsilon$ and $U$ respectively. It is not hard to see that the existence of the variational derivative of $R_\epsilon$ in the sense of \eqref{variation derivative} is implied by the existence of $\theta_Y R$. From equation \eqref{Var of R}, the existence of $\theta_Y R=\lim_{\epsilon\to 0}\frac{\varphi_\epsilon^*R_ \epsilon-R_0}{\epsilon}$ in the weak-$*$-sense is equivalent to the existence of the standard derivative of $r(\epsilon)=\langle u , \varphi_{\epsilon}^{*} R_ \epsilon \varphi_{-\epsilon}^{*} v \rangle$ with respect to $\epsilon$ for all $u, v\in C_0^\infty(U)$.

The following theorem is the well known Hadamard variation formula for the resolvent extended to case of unbounded domains in our setting.
\begin{theorem}
	\label{existence thm}
Let $\varphi_\epsilon$ be a boundary translation flow, then the variational derivative of $R_{\lambda,\epsilon}$ exists in the weak-$*$ topology. Let $G_{\lambda,\epsilon}(x,y)$ be the kernel of $R_{\lambda,\epsilon}$, then its variational derivative is given by
\begin{equation}
\label{Hadamard variation for resolvent}
\delta_Y G_{\lambda,0}(x,y)=\int_{\partial U} \partial'_\nu G_{\lambda,0}(x,z) \partial_\nu G_{\lambda,0}(z,y)\langle Y,\nu\rangle\dd \sigma(z) .
\end{equation}
\end{theorem}
\begin{proof}
Firstly, we prove the existence of $\theta_Y R_\lambda$ in weak-$*$-sense. We know that, from equation \eqref{Rrel}, $\varphi_{\epsilon}^{*} R_{\lambda,\epsilon} \varphi_{-\epsilon}^{*}-R_{\free,\lambda}=-\varphi_{\epsilon}^{*} \calS_{\lambda,\epsilon} Q_{\lambda,\epsilon}^{-1}\calS_{\lambda,\epsilon}^t \varphi_{-\epsilon}^{*}$. We will therefore establish differentiability of 
\begin{equation}
\label{pullback of R}
\begin{aligned}
-  \langle g, \varphi_{\epsilon}^{*}\calS_{\lambda,\epsilon} Q_{\lambda,\epsilon}^{-1}\calS_{\lambda,\epsilon}^t \varphi_{-\epsilon}^{*} f \rangle
= -\langle g,\left( \varphi_{\epsilon}^{*}\calS_{\lambda,\epsilon} \varphi_{-\epsilon}^{*} \right) \left(\varphi_{\epsilon}^{*} Q_{\lambda,\epsilon}^{-1} \varphi_{-\epsilon}^{*}\right) \left(\varphi_{\epsilon}^{*}\calS_{\lambda,\epsilon}^t \varphi_{-\epsilon}^{*}\right) f\rangle
\end{aligned}
\end{equation}
for any fixed test functions $g,f \in C^\infty_0(U)$ and compute its derivative. 
In the last term of equation \eqref{pullback of R}, the operator $\calS_{\lambda,\epsilon}^t $ is the transpose operator to $\calS_{\lambda,\epsilon} $ obtained from the real inner product, i.e. $\calS_{\lambda,\epsilon}^t f = \overline{\calS_{\lambda,\epsilon}^t \overline{f}}$.  Since the free resolvent is smooth off the diagonal and
\begin{align*}
\langle g, \varphi_{\epsilon}^{*}\calS_{\lambda,\epsilon} \varphi_{-\epsilon}^* f \rangle
&=\int_{U}\int_{\partial U_\epsilon} g(x)G_{\free,\lambda}(\varphi_\epsilon(x),\tilde{y})f(\varphi_{-\epsilon}(\tilde{y})) \dd \sigma(\tilde{y}) \dd x
\\
&=\int_{U}\int_{\partial U} g(x) G_{\free,\lambda}(\varphi_\epsilon(x),\varphi_{\epsilon}(y))f(y) \dd \sigma(y) \dd x ,
\end{align*}
the kernel of $\calS_{\lambda,\epsilon}$ is smooth on $U \times \partial U$ .
Here $f \in C^\infty(\partial U)$, $g \in C_0^\infty(U)$ and $G_{\free,\lambda}$ is the kernel of the free resolvent. 
To establish differentiability by the product rule it is sufficient to prove the existence of $\frac{\dd}{\dd \epsilon}\left( \varphi_{\epsilon}^{*}\calS_{\lambda,\epsilon} \varphi_{-\epsilon}^{*} \right)$  in the $C^\infty$-topology of integral kernels on $U \times \partial U$, and the existence of $\frac{\dd}{\dd \epsilon}\left(\varphi_{\epsilon}^{*} Q_{\lambda,\epsilon}^{-1} \varphi_{-\epsilon}^{*}\right)$  at $\epsilon=0$ in the weak-$*$-sense. 
The free resolvent kernel is smooth off the diagonal and therefore the above formula for $\varphi_{\epsilon}^{*}\calS_{\lambda,\epsilon} \varphi_{-\epsilon}^*$ shows differentiability of the smooth kernel in the parameter $\epsilon$ at $\epsilon=0$ and the classical sense.
We are thus left with proving the existence of $\frac{\dd}{\dd \epsilon}\left(\varphi_{\epsilon}^{*} Q_{\lambda,\epsilon}^{-1} \varphi_{-\epsilon}^{*}\right)$ at $\epsilon=0$ in the weak-$*$-sense. Note that $\varphi_{\epsilon}^{*} Q_{\lambda,\epsilon}^{-1} \varphi_{-\epsilon}^{*}$ is a one-parameter family of maps from $C^\infty(\partial U)$ to $C^\infty(\partial U)$, i.e. the spaces do not depend on $\epsilon$. Similar to equation \eqref{Q decomp} we have the splitting
\begin{equation*}
\tilde{Q}_\lambda=\sum_{j=1}^N p_jQ_\lambda p_j, \quad
T_\lambda=\sum_{j\ne k}^N p_jQ_\lambda p_k \quad \text{and} \quad
Q_\lambda=\tilde{Q}_\lambda+T_\lambda,
\end{equation*}
where $p_j$ are the orthogonal projections $L^2(\partial U) \to L^2(\partial U_j)$ and $\partial U_j, j=1,\ldots,N$ are the connected components of $\partial U$.
Define $\calQ_{\lambda,\epsilon}=\varphi_{\epsilon}^* Q_{\lambda,\epsilon} \varphi_{-\epsilon}^*$, $\tilde{\calQ}_{\lambda,\epsilon}=\varphi_{\epsilon}^* \tilde{Q}_{\lambda,\epsilon} \varphi_{-\epsilon}^*$ and $\mathcal{T}_{\lambda,\epsilon}=\varphi_{\epsilon}^* T_{\lambda,\epsilon} \varphi_{-\epsilon}^*$. By the definition of $\tilde{Q}_{\lambda,\epsilon}$, we have
$$
\tilde{\calQ}_{\lambda,\epsilon}f(x_i)=\int_{\partial U_i} G_{\free,\lambda}(\varphi_\epsilon (x_i),\varphi_\epsilon (y_i))f_i (y_i) \varphi_\epsilon^*(\dd \sigma(y_i)),
$$
where $f\in C^\infty(\partial U)$ and $f_i=p_i f$. As $\varphi_\epsilon$ is a boundary translation flow, we obtain the following relationships.
$$
\tilde{\calQ}_{\lambda,\epsilon}=\tilde{Q}_{\lambda,0}=\tilde{Q}_\lambda \qquad \text{and} \qquad \tilde{\calQ}_{\lambda,\epsilon}^{-1}=\tilde{Q}_\lambda^{-1} .
$$
Now, from the decomposition of $Q_\lambda$ in equation \eqref{Q decomp}, one obtains 
\begin{equation*}
\left\{
\begin{aligned}
&\frac{\dd}{\dd\epsilon}\calQ_{\lambda,\epsilon}=\frac{\dd}{\dd\epsilon}(\tilde{\calQ}_{\lambda,\epsilon}+\mathcal{T}_{\lambda,\epsilon})=\frac{\dd}{\dd\epsilon}\mathcal{T}_{\lambda,\epsilon}, \\
&\frac{\dd}{\dd\epsilon} \calQ_{\lambda,\epsilon}^{-1}=-\calQ_{\lambda,\epsilon}^{-1} (\frac{\dd}{\dd\epsilon}\calQ_{\lambda,\epsilon}) \calQ_{\lambda,\epsilon}^{-1}.
\end{aligned}
\right.
\end{equation*}

The family $\mathcal{T}_{\lambda,\epsilon}$ is a differentiable family of smoothing operators for sufficiently small $\epsilon$ (i.e. no obstacles are overlapping) and its derivative in $\epsilon$ therefore, by Taylor's remainder estimate, exits as a family of smooth kernels. Hence, $\calQ_{\lambda,\epsilon}$ is differentiable in $\epsilon$ at $\epsilon=0$
as a family of operators from $H^s(\partial U)$ to $H^{s+1}(\partial U)$ for any $s \in \R$.
We now use that $\calQ_{\lambda,0}$ is invertible and 
the inverse $\calQ_{\lambda,0}^{-1}$ is a pseudodifferential operator of order one, and maps $H^s(\partial U)$ to $H^{s-1}(\partial U)$ . 
Since the space of invertible operators is open the inverses $\calQ_{\lambda,\epsilon}^{-1}$ exist near $\epsilon=0$ as maps from $H^s(\partial U)$ to $H^{s-1}(\partial U)$.
Hence, $\calQ_{\lambda,\epsilon}^{-1}$ is differentiable in $\epsilon$ at $\epsilon=0$ as a family of operators from $H^s(\partial U)$ to $H^{s-1}(\partial U)$. In particular the derivatives $\frac{\dd}{\dd\epsilon}\calQ_{\lambda,\epsilon}$ and $\frac{\dd}{\dd\epsilon} \calQ_{\lambda,\epsilon}^{-1}$ exist in the weak-$*$-sense.
Hence the variational derivative of $R_{\lambda,\epsilon}$ exists in the weak-$*$ sense and it is given by 
\begin{equation}
\label{dot R}
\delta_Y R_\lambda=\theta_Y R_\lambda- YR_\lambda- Y' R_\lambda,
\end{equation}
where $Y'$ means the action of $Y$ on the second variable.
It remains to compute the derivative.
To do this we consider the inhomogeneous problem
\begin{equation}
\label{inhomogeneous problem}
E_\epsilon(u_\epsilon,\varphi_{-\epsilon}^{*}v)=\int_{U_\epsilon} f_\epsilon (\varphi_{-\epsilon}^{*}v) \dd x.
\end{equation}
Let $\nu$ be the exterior unit normal of $U$, $y$ be an interior point of $U$ such that $\varphi_{\epsilon}(y)=\tilde{y}$ and $e(u,v)=\nabla u\cdot\nabla v -\lambda^2 uv$. By applying
$$
\left\{
\begin{aligned}
&u_\epsilon(x)=G_{\lambda,\epsilon}(x,\varphi_{\epsilon}(y)) ,\\
&f_\epsilon(x)=\delta_{\varphi_{\epsilon}(y)}(x).
\end{aligned}
\right.
$$
to equation \eqref{inhomogeneous problem}, taking derivative in $\epsilon$ of equation \eqref{inhomogeneous problem} and using equation \eqref{resolvent equation} and Peetre's computations \cite{Peetre}, one has
\begin{equation}
\label{Hadamard eqn1}
\int_{U}e(\delta_Y {u}_0,v)\dd x
=
\int_{U}e(u_0,Y(v))\dd x-\int_{\partial U} e(u_0,v) \langle Y,\nu\rangle \dd\sigma.
\end{equation}
Let $v(x)=G_{\lambda,0}(x,z)$, we have
\begin{equation}
\label{Hadamard eqn3}
\int_U e(u,v) \dd x=u(z)+\int_{\partial U}u(x)\partial_\nu G_{\lambda,0}(x,z) \dd\sigma(x)\,.
\end{equation}
Using the symmetric property of $G_{\lambda,\epsilon}(x,y)=G_{\lambda,\epsilon}(y,x)$ and equations \eqref{dot R}, \eqref{Hadamard eqn1} and \eqref{Hadamard eqn3}, we obtain
\begin{equation*}
\delta_Y {G}_{\lambda,0}(z,y)+\int_{\partial U}\partial'_\nu G_{\lambda,0}(z,x) \delta_Y {G}_{\lambda,0}(x,y) \dd\sigma(x)=0.
\end{equation*}
Using the boundary conditions $G_{\lambda,\epsilon}(\varphi_\epsilon(x),y)=0$ for $x \in\partial U$, we arrive at the Hadamard variation formula for the Dirichlet resolvent.
\begin{equation*}
\delta_Y {G}_{\lambda,0}(z,y)=\int_{\partial U} \partial'_\nu G_{\lambda,0}(z,x) \partial_\nu G_{\lambda,0}(x,y)\langle Y,\nu\rangle\dd\sigma(x) .
\end{equation*}
\end{proof}

\subsection{Application of the Hadamard variation formula to the relative resolvent}
We now apply the Hadamard variation formula to our setting with finitely many obstacles, combining the above formulae for $U=\mathcal{E}$ and $U=\calO$.
 We have from equation \eqref{Hadamard variation for resolvent}
\begin{equation}
\label{Hadamard variation for resolvent in and out}
\delta_Y {G}_{\lambda,0}(z,y)=\left\{
\begin{aligned}
&\int_{\partial X}\left[\left(\partial_\nu G_{\lambda,0}(x,y)\partial_\nu G_{\lambda,0}(x,z)\right)\langle Y,\nu_X\rangle\right]_+ \dd \sigma(x)& &y,z \in \calE \\
&\int_{\partial X}\left[\left(\partial_\nu G_{\lambda,0}(x,y)\partial_\nu G_{\lambda,0}(x,z)\right)\langle Y,\nu_X\rangle\right]_- \dd \sigma(x)& &y,z \in \calO
\end{aligned}
\right.
\end{equation}
where $(\cdot)_+$ means taking limits from $\calE$ to the boundary $\partial\calE$ and $(\cdot)_-$ means taking limits from $\calO$ to the boundary $\partial\calO$. We will now use the variational formula for the relative resolvent to prove the following theorem. Hence, we define $\calO_\epsilon= \varphi_\epsilon(\calO)$ and similarly $\calO_{j,\epsilon}= \varphi_\epsilon(\calO_j)$. In this way we can define the relative resolvent
$$
 R_{\rel,\lambda,\epsilon} = R_{\calO_\epsilon,\lambda} -  R_{\free,\lambda} - \left( \sum_{j=1}^N R_{\calO_{j,\epsilon},\lambda} -  R_{\free,\lambda}  \right)
$$  
and its integral kernel  $G_{\rel,\lambda,\epsilon}$ depending on the parameter $\epsilon$. 

\begin{theorem}
	\label{trace thm}
Let $\varphi_\epsilon$ be a boundary translation flow, $d\ge 2$ and $\lambda \in \mathfrak{D}_\nu$. Then $R_{\rel,\lambda,\epsilon}$ is for each $\lambda \in \mathfrak{D}_\nu$ a $C^{1}$ trace-class operator valued function of $\epsilon$ near the point $\epsilon=0$. Its derivative $\dot{R}_{\rel,\lambda,0}$ equals $\delta_Y R_{\rel,\lambda}$ and there exists $\delta>0$ such that $L^2$-trace-norm of $\dot{R}_{\rel,\lambda,0}$ is bounded by
\begin{equation}
\label{trace thm norm bound}
\|\dot{R}_{\rel,\lambda,0}\|_1 \le C_\nu \rho(\Im \lambda)e^{-\delta \Im \lambda}, \quad \lambda \in \mathfrak{D}_\nu.
\end{equation}
Its kernel, $\dot{G}_{\rel,\lambda,0}=\breve{\dot{R}}_{\rel,\lambda,0}$, is given by
\begin{equation}
\label{variation Rrel 1}
\begin{aligned}
\dot{G}_{\rel,\lambda,0}(x,y)=&\sum_{i=1}^N\int_{\partial \calO_i} \Big[\big(\partial'_\nu G_{\calO,\lambda,0}(x,z)\partial_\nu G_{\calO,\lambda,0}(z,y) 
\\
&-\partial'_\nu G_{\calO_i,\lambda,0}(x,z)\partial_\nu G_{\calO_i,\lambda,0}(z,y)\big)\langle Y,\nu_{\calE_i} \rangle \Big]_+ \dd\sigma(z)
\end{aligned}
\end{equation}
for $x,y \in \calE$ or $x,y \in \calO$, where $\nu_{\calE_i}$ is the exterior unit normal of $\calE_i$ and $\calE_i = \R^d \setminus \overline{\calO_i}$. 
\end{theorem}
\begin{proof}
We start by showing that the family is Fr\'echet differentiable in the Banach space of trace-class operators with continuous derivative, i.e. the function is $C^1$ as a trace-class operator valued function on $I$, where is a fixed sufficiently small compact interval.
As in the previous section we have $R_{\rel,\lambda}=-\calS_\lambda (Q_\lambda^{-1}-\tilde{Q}_\lambda^{-1}) \calS_\lambda^t$. We can decompose its variation form as in the proof of Theorem \ref{existence thm}
$$
R_{\rel,\lambda,\epsilon}=-\calS_{\lambda,\epsilon} (Q_{\lambda,\epsilon}^{-1}-\tilde{Q}_{\lambda,\epsilon}^{-1}) \calS_{\lambda,\epsilon}^t=- \calS_{\lambda,\epsilon} \varphi_{-\epsilon}^* \varphi_{\epsilon}^* (Q_{\lambda,\epsilon}^{-1}-\tilde{Q}_{\lambda,\epsilon}^{-1}) \varphi_{-\epsilon}^* \varphi_{\epsilon}^* \calS_{\lambda,\epsilon}^t .
$$
Then we split the last term into a product of three terms, i.e. $\calS_{\lambda,\epsilon} \varphi_{-\epsilon}^*$, $\varphi_{\epsilon}^* (Q_{\lambda,\epsilon}^{-1}-\tilde{Q}_{\lambda,\epsilon}^{-1}) \varphi_{-\epsilon}^*$, and $\varphi_{\epsilon}^* \calS_{\lambda,\epsilon}^t $. The first operator $\calS_{\lambda,\epsilon} \varphi_{-\epsilon}^*$ is given by 
$$
\calS_{\lambda,\epsilon} \varphi_{-\epsilon}^* f (x)=\int_{\partial \calO_\epsilon} G_{\free,\lambda}(x,\tilde{y})f(\varphi_{-\epsilon}(\tilde{y})) \dd {\sigma(\tilde y)}=\int_{\partial \calO} G_{\free,\lambda}(x,\varphi_{\epsilon}(y))f(y) \varphi_{\epsilon}^* (\dd \sigma(y)).
$$
Since $\partial\calO$ is a disjoint union of the components $\partial \calO_j$ the operators $\calS_{\lambda,\epsilon} \varphi_{-\epsilon}^*$ splits into a sum $\sum_j \calS_{j,\lambda,\epsilon} = \sum_j T_j(\epsilon) \calS_{j,\lambda}$, where $\calS_{j,\lambda}: L^2(\partial \calO_j) \to H^1(\R^d)$ and $T_j(\epsilon): L^2(\R^d) \to L^2(\R^d)$ is the translation $T_j(\epsilon) f(x) = f(x - Z_j \epsilon)$. Here we used the fact that $Z$ is constant and equal to $Z_j$ near $\partial \calO_j$ and that the free Green's function is translation invariant. Since $T_j(\epsilon)$
is $C^1$ as a family of maps $H^1(\R^d) \to L^2(\R^d)$ this shows that $\calS_{\lambda,\epsilon} \varphi_{-\epsilon}^*$ and its adjoint are $C^1$ as families of bounded operators $L^2(\partial \calO) \to L^2(\R^d)$. As shown in the proof of Theorem \ref{existence thm}, the operator $\varphi_{\epsilon}^*\tilde Q_{\lambda,\epsilon} \varphi_{-\epsilon}^*$ is independent of $\epsilon$ and therefore
$$
\varphi_{\epsilon}^* Q_{\lambda,\epsilon} \varphi_{-\epsilon}^* =  \tilde Q_{\lambda} + \varphi_{\epsilon}^* T_{\lambda,\epsilon} \varphi_{-\epsilon}^*=\tilde Q_{\lambda}+\mathcal{T}_{\lambda,\epsilon}.
$$
The map $ \mathcal{T}_{\lambda,\epsilon}$ has smooth integral kernel that depends smoothly on $\epsilon$ for sufficiently small $\epsilon$. This family is therefore $C^1$ as a family of trace-class operators. We temporarily denote $\varphi_{\epsilon}^* Q_{\lambda,\epsilon}  \varphi_{-\epsilon}^*$
by $J_\epsilon$, As $G_{\free,\lambda}$ is translation invariant, $\varphi_{\epsilon}^* \tilde Q_{\lambda,\epsilon}  \varphi_{-\epsilon}^*$ is independent of $\epsilon$ and hence we denote it by $\tilde J$. Then, by the above
$J_\epsilon = J_0 + r_\epsilon$ and $J_0 - \tilde J$ is trace-class. Moreover, the remainder term $r_\epsilon$ is of the form
$r_\epsilon = \dot J_0 \epsilon + \rho(\epsilon)$, where $\| \rho(\epsilon) \|_1 = o(\epsilon)$ as $\epsilon \to 0$. By the Neumann series we have
$$
 J_\epsilon^{-1} - \tilde J^{-1} =  (1 + J_0^{-1} r_\epsilon)^{-1} J_0^{-1} - \tilde J^{-1}= J_0^{-1} - \tilde J^{-1} - J_0^{-1} \dot J_0 J_0^{-1} \epsilon + \tilde \rho(\epsilon),
$$
where again $\| \tilde \rho(\epsilon) \|_1 = o(\epsilon)$. We have used here that trace-class operators form an ideal in the algebra of bounded operators, and the norm estimate $\| A B \|_1 \leq \| A \| \| B \|_1$ holds.
We conclude that $\varphi_{\epsilon}^* (Q_{\lambda,\epsilon}^{-1}-\tilde{Q}_{\lambda,\epsilon}^{-1}) \varphi_{-\epsilon}^*$ is $C^1$ as a family of trace-class operators. 

We now compute the derivative in $\epsilon$ of all the three terms. We obtain
\begin{align*}
	\dot R_{\rel,\lambda,0}=&\calS_{\lambda} Y (Q_{\lambda}^{-1}-\tilde{Q}_{\lambda}^{-1}) \calS_{\lambda}^t
	\\
	&-\calS_{\lambda}\left.\frac{\dd}{\dd \epsilon}\left(\varphi_{\epsilon}^*  (Q_{\lambda,\epsilon}^{-1} -\tilde{Q}_{\lambda,\epsilon}^{-1} ) \varphi_{-\epsilon}^*\right)\right|_{\epsilon=0}\calS_{\lambda}^t
	-\calS_{\lambda} (Q_{\lambda}^{-1}-\tilde{Q}_{\lambda}^{-1}) Y \calS_{\lambda}^t .
\end{align*}
For the derivative of the second term, we have as in Theorem \ref{existence thm}
$$
\tilde{\calQ}_{\lambda,\epsilon}f(x_i)=\int_{\partial \calO_i} G_{\free,\lambda}(\varphi_\epsilon (x_i),\varphi_\epsilon (y_i))f_i (y_i) \varphi_\epsilon^*(\dd \sigma(y_i)),
$$
and
$$
\left.\frac{\dd}{\dd \epsilon}\left(\varphi_{\epsilon}^*  (Q_{\lambda,\epsilon}^{-1} -\tilde{Q}_{\lambda,\epsilon}^{-1} ) \varphi_{-\epsilon}^*\right)\right|_{\epsilon=0}=\left.\frac{\dd}{\dd\epsilon}\calQ_{\lambda,\epsilon}^{-1}\right|_{\epsilon=0}=-\left.\left(\calQ_{\lambda,\epsilon}^{-1} \left(\frac{\dd}{\dd\epsilon}\mathcal{T}_{\lambda,\epsilon}\right) \calQ_{\lambda,\epsilon}^{-1}\right)\right|_{\epsilon=0}.
$$

Therefore, the variation of the relative resolvent is given by
\begin{equation}
\label{Rrel dot eqn}
\begin{aligned}
	\dot R_{\rel,\lambda,0} &=\calS_{\lambda} Y (Q_{\lambda}^{-1}-\tilde{Q}_{\lambda}^{-1}) \calS_{\lambda}^t
	-\calS_{\lambda}Q_\lambda^{-1}\left.\frac{\dd}{\dd\epsilon}\mathcal{T}_{\lambda,\epsilon}\right|_{\epsilon=0}Q_\lambda^{-1}\calS_{\lambda}^t
	-\calS_{\lambda} (Q_{\lambda}^{-1}-\tilde{Q}_{\lambda}^{-1}) Y \calS_{\lambda}^t
	\\
	&=\calS_{\lambda} Y (Q_{\lambda}^{-1}T_\lambda\tilde{Q}_{\lambda}^{-1}) \calS_{\lambda}^t
	-\calS_{\lambda}Q_\lambda^{-1}\dot{\mathcal{T}}_{\lambda,0}Q_\lambda^{-1} \calS_{\lambda}^t
	-\calS_{\lambda} (Q_{\lambda}^{-1}T_\lambda\tilde{Q}_{\lambda}^{-1}) Y \calS_{\lambda}^t.
\end{aligned}
\end{equation}
To estimate the trace-norm of $\dot R_{\rel,\lambda,0}$, note that the first term in the above equation can be estimated by
\begin{equation*}
	\|\calS_{\lambda} Y (Q_{\lambda}^{-1}T_\lambda\tilde{Q}_{\lambda}^{-1}) \calS_{\lambda}^t\|_1\le \|\calS_{\lambda} Y Q_{\lambda}^{-1}\|_{H^{\frac{1}{2}}\to L^2} \|E_{-1}\|_{H^{-\frac{1}{2}}\to H^{\frac{1}{2}}} \|E_{1} T_\lambda\|_1 \|\tilde{Q}_{\lambda}^{-1} \calS_{\lambda}^t\|_{L^2\to H^{-\frac{1}{2}}} ,
\end{equation*}
where $E_s=(\sqrt{-\Delta_{\partial\calO}}+1)^s$ and $\|E_{1} T_\lambda\|_1 $ is the trace-norm from $H^{-\frac{1}{2}} \to H^{-\frac{1}{2}}$. Now by Proposition \ref{S properties}, we have
\begin{equation}
	\label{Rrel dot eqn 1}
	\|\calS_{\lambda} Y (Q_{\lambda}^{-1}T_\lambda\tilde{Q}_{\lambda}^{-1}) \calS_{\lambda}^t\|_1\le C (1+(\Im \lambda)^2)^4 \rho(\Im \lambda) \|E_{1} T_\lambda\|_1.
\end{equation}
The third term in equation \eqref{Rrel dot eqn} can be bounded the same as the first term. For the second term one can estimate it by
\begin{equation}
\label{Rrel dot eqn 2}
\begin{aligned}
	\|\calS_{\lambda}Q_\lambda^{-1}\dot{\mathcal{T}}_{\lambda,0}Q_\lambda^{-1} \calS_{\lambda}^t\|_1 &\le \|\calS_{\lambda}Q_{\lambda}^{-1}\|_{H^{\frac{1}{2}}\to L^2} \|E_{-1}\|_{H^{-\frac{1}{2}}\to H^{\frac{1}{2}}} \|E_{1} \dot{\mathcal{T}}_{\lambda,0}\|_1 \|{Q}_{\lambda}^{-1} \calS_{\lambda}^t\|_{L^2\to H^{-\frac{1}{2}}}
	\\
	&\le C (1+(\Im \lambda)^2)^4 \rho(\Im \lambda)\|E_{1} \dot{\mathcal{T}}_{\lambda,0}\|_1.
\end{aligned}
\end{equation}
Combining equations \eqref{Rrel dot eqn 1} and \eqref{Rrel dot eqn 2}, one has
\begin{equation*}
	\|\dot R_{\rel,\lambda,0}\|_1\le C (1+(\Im \lambda)^2)^4 \rho(\Im \lambda) \left(\|E_{1} T_\lambda\|_1+\|E_{1} \dot{\mathcal{T}}_{\lambda,0}\|_1\right).
\end{equation*}
In order to prove the bound \eqref{trace thm norm bound}, it suffices to prove that 
\begin{equation*}
	\|\mathcal{T}_{\lambda,\epsilon}\|_1 \le C e^{-\delta \Im \lambda}  \quad \text{and} \quad \|\dot{\mathcal{T}}_{\lambda,\epsilon}\|_1 \le C e^{-\delta \Im \lambda}.
\end{equation*}
As in Theorem \ref{existence thm}, $\mathcal{T}_{\lambda,\epsilon}$ is a smoothing operator that depends smoothly on $\epsilon$, as long as we have $\dist(\partial\calO_{j,\epsilon},\partial\calO_{k,\epsilon})>0$ for all pairs of obstacles. Since the obstacles are compact, $\mathcal{T}_{\lambda,\epsilon}$ is a smoothing operator on compact domains and hence it is also a trace-class operator from $H^s(\partial \calO) \to H^s(\partial \calO)$ for all $s \in \R$. This proves the first part of the theorem.
	
Also, from Theorem \ref{existence thm}, we know that $\delta_Y G_{\rel,\lambda,0}$ exists in the weak-$*$ sense. By equations \eqref{pullback of R} and \eqref{dot R}, we know that the kernel of $\dot R_{\rel,\lambda,0}$ coincides with $\delta_Y G_{\rel,\lambda,0}$. In other words, the variational derivative exists in a stronger sense. We can therefore apply the variation formula  \eqref{Hadamard variation for resolvent in and out} to the relative resolvent, which gives
\begin{align*}
\dot R_{\rel,\lambda,0}(x,y)=&\delta_Y G_{\rel,\lambda,0}(x,y)
\\ 
=&\delta_Y\left(G_{\calO,\lambda,0}-G_\free-\sum_{i=1}^N (G_{\calO_i,\lambda,0}-G_\free)\right)(x,y)
\\
=&\delta_Y G_{\calO,\lambda,0}(x,y)-\sum_{i=1}^N \delta_Y G_{\calO_i,\lambda,0}(x,y;z)
\\
=&\sum_{i=1}^N\int_{\partial \calO_i} \left[\partial'_\nu G_{\calO,\lambda,0}(x, q)\partial_\nu G_{\calO,\lambda,0}( q,y) \langle Y,\nu_{\calE_i} \rangle \right]_+ \dd \sigma(q) 
\\
&-\sum_{i=1}^N\int_{\partial \calO_i} \left[\partial'_\nu G_{\calO_i,\lambda,0}(x, q)\partial_\nu G_{\calO_i,\lambda,0}( q,y) \langle Y,\nu_{\calE_i} \rangle \right]_+ \dd \sigma(q) 
\\
=&\sum_{i=1}^N\int_{\partial \calO_i} \Big[\big(\partial'_\nu G_{\calO,\lambda,0}(x, q)\partial_\nu G_{\calO,\lambda,0}( q,y) 
\\
&-\partial'_\nu G_{\calO_i,\lambda,0}(x, q)\partial_\nu G_{\calO_i,\lambda,0}( q,y)\big)\langle Y,\nu_{\calE_i} \rangle \Big]_+ \dd \sigma(q)  .
\end{align*}

Since the interior parts of $G_{\calO,\lambda,\epsilon}$ and $G_{\calO_j,\lambda,\epsilon}$ are the same, the interior contributions of $\delta_Y G_{\calO,\lambda,0}$ in the expression \eqref{Hadamard variation for resolvent in and out} cancel out with the ones of $\delta_Y G_{\calO_j,\lambda,0}$. Therefore, we are left with only the exterior contributions as shown in equation \eqref{variation Rrel 1}.

\end{proof}
\begin{definition}
	Let $M$ be a Riemannian manifold with  $A$ be an operator on $L^2(M)$ with continuous kernel Schwartz kernel  $\breve A\in C(M \times M)$. For an open subset $V \subset M$, we define the localised trace on $V$ as
	\begin{equation*}
	\tr_V(A)=\int_V \breve A(x,x) \dd \mathrm{Vol}_g(x),
	\end{equation*}
	whenever the integral exists. If $A$ is trace-class and has continuous kernel the trace of $A$ on $L^2(M)$ then equals $\tr_M(A)$ by Mercer's theorem. We also write $\tr_V(\breve A)$.
\end{definition}
Our next proposition gives a relationship between the trace of the variation of the relative resolvent with a local trace on the boundary.
\begin{proposition}
	\label{trace variation prop}
	The trace of the variation of the relative resolvent on $L^2(X)$ is also given by 
		\begin{equation*}
	\label{Trace variation}
	\Tr_{X}[\delta_Y R_{\rel,\lambda}]=\frac{1}{2\lambda}\sum_{i=1}^N\Tr_{\partial \calO_i}\left[ \left(\partial_\nu\partial_\nu'\frac{\partial  G_{\calO,\lambda}}{\partial \lambda}-\partial_\nu\partial_\nu'\frac{\partial G_{\calO_i,\lambda}}{\partial \lambda}\right)\langle Y,\nu\rangle \right]_+.
	\end{equation*}
\end{proposition}
\begin{proof}
Let $\gamma_{\calE,\nu} : H^{\frac{3}{2}}(\calE)\to L^2(\partial \calO)$ be the Sobolev trace after taking the exterior normal derivative ($\nu$ is the exterior normal vector field) and 
$\calB_{\calO,\lambda}: L^2(\partial \calO) \to L^2(\R^d)$
is defined as $\calB_{\calO,\lambda}=  R_{\calO,\lambda} \circ \gamma_{\calE,\nu}^*$. To see that this map is well defined we note that
\begin{gather} \label{Bformula}
  R_{\calO,\lambda}\gamma_{\calE,\nu}^* =  R_{\free,\lambda} \gamma_{\calE,\nu}^* - \calS_{\calO,\lambda} Q_{\calO,\lambda}^{-1} \calS_{\calO,\lambda}^t \gamma_{\calE,\nu}^*,
\end{gather}
where $\calS_{\calO,\lambda}$ and $Q_{\calO,\lambda}$ are the same as $\calS_\lambda$ and $Q_\lambda$ defined in \eqref{Rdiff}, but with emphasis on the dependence on $\calO$. Then $R_{\free,\lambda} \gamma_{\calE,\nu}^*$ maps $L^2(\partial \Omega)$ continuously to $L^2(\R^d)$. The operator $\calS_{\lambda}^t \gamma_{\calE,\nu}^*$ is the double layer operator on the boundary $\partial \calO$ and maps continuously $L^2(\partial \Omega) \to L^2(\partial \Omega)$ (see for example \cite{ChandlerWilde}). Since $Q_\lambda^{-1}$ is a pseudodifferential operator of order one, and $ \calS_{\lambda}$ continuously maps $H^{-1}(\partial \calO)$ to $H^\frac{1}{2}(\R^d) \subset L^2(\R^d)$ we see that  
$\calB_{\calO,\lambda}: L^2(\partial \calO) \to L^2(\R^d)$ is indeed well defined and continuous.

Similarly, we define $\gamma_{\calO_i,\nu}$ and $\calB_{\calO_i,\lambda}$. Let $p_{i}$ be the orthogonal projection $L^2(\partial \calO) \to L^2(\partial \calO_i)$. Then from \eqref{Bformula} then have for $(\calB_{\calO,\lambda}-\calB_{\calO_i,\lambda}) p_{i}$ the representation
\begin{align*}
\left( \calB_{\calO,\lambda}-\calB_{\calO_i,\lambda} \right) p_{i} &= S_\lambda \left( Q_\lambda^{-1} - Q_{\calO_i,\lambda}^{-1} p_i  \right)S_\lambda^t  \gamma_{\calE,\nu}^* p_i
=S_\lambda \left( Q_\lambda^{-1} - \tilde Q_{\lambda}^{-1} p_i  \right)S_\lambda^t  \gamma_{\calE,\nu}^* p_i\\
&=S_\lambda \left( Q_\lambda^{-1} - \tilde Q_{\lambda}^{-1}  \right) S_\lambda^t  \gamma_{\calE,\nu}^* p_i + \sum_{k \not= i} S_\lambda \tilde Q_{\lambda}^{-1} p_k S_\lambda^t  \gamma_{\calE,\nu}^* p_i. 
\end{align*}
where $Q_{\calO_i,\lambda}$ is the operator $Q_\lambda$ in equation \eqref{Rdiff} when $\calO$ is replaced by $\calO_i$ in the definition. As in \eqref{Q decomp} we have used the decomposition $Q_\lambda = \tilde Q_\lambda + T_\lambda$. Since 
$T_\lambda$ is smoothing, so is $Q_\lambda^{-1}-\tilde Q_\lambda^{-1} = - Q_\lambda^{-1} T_\lambda \tilde Q_\lambda^{-1}$. Similarly, as the free Green's function is smooth off the diagonal the operator $p_k S_\lambda^t  \gamma_{\calE,\nu}^* p_i$ has smooth integral kernel for $k \not= i$. In particular these operators are trace-class as maps from $L^2(\partial \calO)$ to $L^2(\partial \calO)$. Since $S_\lambda$ as well as $S_\lambda \tilde Q_{\lambda}^{-1}$ is bounded from $L^2(\partial \calO)$ to $L^2(\R^d)$ this shows that for every $i \in \{1,\ldots, N\}$ the operator $(\calB_{\calO,\lambda}-\calB_{\calO_i,\lambda}) p_{i}$ is nuclear.

Equation \eqref{variation Rrel 1} can be rewritten as
\begin{equation*}
\delta_Y R_{\rel,\lambda}=\sum_{i=1}^N \left( \calB_{\calO,\lambda} Y_i \calB_{\calO,\lambda}^* - \calB_{\calO_i,\lambda} Y_i \calB_{\calO_i,\lambda}^* \right),
\end{equation*}
where $Y_i=\langle Y,\nu_{\calE_i} \rangle |_{\partial \calO_i}$ is viewed as a multiplication operator acting on $L^2(\partial \calO_i) \subset L^2(\partial \calO)$. 
Taking the trace, we have
\begin{align*}
\Tr(\delta_Y R_{\rel,\lambda})=&\sum_{i=1}^N\Tr\left(  \calB_{\calO,\lambda} Y_i \calB_{\calO,\lambda}^* - \calB_{\calO_i,\lambda} Y_i \calB_{\calO_i,\lambda}^* \right)
\\
=&\sum_{i=1}^N\Tr \left(  \left(\calB_{\calO,\lambda}-\calB_{\calO_i,\lambda} \right) Y_i \calB_{\calO,\lambda}^* + \calB_{\calO_i,\lambda} Y_i \left(\calB_{\calO,\lambda}^*- \calB_{\calO_i,\lambda}^*\right) \right)
\\
=&\sum_{i=1}^N\Tr \left(\calB_{\calO,\lambda}^*\left(\calB_{\calO,\lambda}  -\calB_{\calO_i,\lambda} \right)Y_i + Y_i \left(\calB_{\calO,\lambda}^*-\calB_{\calO_i,\lambda}^*\right)\calB_{\calO_i,\lambda}  \right)
\\
=&\sum_{i=1}^N\Tr \left(  \left(\calB_{\calO,\lambda}^* \calB_{\calO,\lambda} -  \calB_{\calO_i,\lambda}^* \calB_{\calO_i,\lambda}  \right)Y_i\right).
\end{align*}
Here the cyclic permutation under the trace is justified because of the nuclearity of $(\calB_{\calO,\lambda}-\calB_{\calO_i,\lambda}) p_i$.
Since $\calB_{\calO,\lambda}^* \calB_{\calO,\lambda} = \frac{1}{2\lambda} \gamma_{\calE,\nu} \frac{\der}{\der \lambda} R_{\calO,\lambda}\gamma_{\calE,\nu}^*$ its integral kernel is $\frac{1}{2\lambda}\partial_\nu\partial_\nu'\frac{\partial G_{\calO,\lambda}}{\partial \lambda}$. We then  obtain
\begin{align*}
\Tr(\delta_Y R_{\rel,\lambda})=&\sum_{i=1}^N\Tr\left( \left( \calB_{\calO,\lambda}^* \calB_{\calO,\lambda} -  \calB_{\calO_i,\lambda}^* \calB_{\calO_i,\lambda}\right)Y_i \right)
\\
=&\frac{1}{2\lambda}\sum_{i=1}^N\Tr_{\partial \calO_i}\left[   \left(\partial_\nu\partial_\nu'\frac{\partial G_{\calO,\lambda}}{\partial \lambda} -  \partial_\nu\partial_\nu'\frac{\partial G_{\calO_i,\lambda}}{\partial \lambda}\right)Y_i \right]_+
\\
=&\frac{1}{2\lambda}\sum_{i=1}^N\Tr_{\partial \calO_i}\left[ \left(\partial_\nu\partial_\nu'\frac{\partial G_{\calO,\lambda}}{\partial \lambda}-\partial_\nu\partial_\nu'\frac{\partial G_{\calO_i,\lambda}}{\partial \lambda}\right)\langle Y,\nu\rangle \right]_+.
\end{align*}
\end{proof}
By Theorem \ref{frel trace thm} one can define trace-class operators $g_\rel=D_f=\frac{\rmi}{\pi}\int_{\tilde{\Gamma}} \lambda f(\lambda) R_{\rel,\lambda} \dd \lambda$ for $f(\lambda)=g(\lambda^2)$ and $g\in {\mathcal{P}}_\theta$. 
We now have the following.
\begin{proposition}
	\label{trace of grel}
Let $g\in {\mathcal{P}}_\theta$ and $\varphi_\epsilon$ be a boundary translation flow. Then $\delta_Y g_\rel$ is a $C^{1}$ trace-class operator valued function of $\epsilon$ near the point $\epsilon=0$. Its derivative $\dot g_\rel = \delta_Y g_\rel$ satisfies
\begin{equation*}
\Tr(\delta_Y g_\rel)=-\sum_{i=1}^N\Tr_{\partial \calO_i}\left[ \partial_\nu\partial_\nu'(\breve g'_\calO-\breve g'_{\calO_i})\langle Y,\nu\rangle\right]_+,
\end{equation*}
where $g'(z)=\frac{\dd g}{\dd z}(z)$.
\end{proposition}
\begin{proof}
By Theorem \ref{trace thm} the operator $R_{\rel}(z)$ is, for fixed $z \nin \mathfrak{S}_\theta$,  in the Banach space $C^1(I, \mathcal{L}_1)$
of trace-class operator valued $C^1$-functions on a compact interval $I$ containing zero. Differentiation defines a closed operator
from on $C(I, \mathcal{L}_1)$ with domain $C^1(I, \mathcal{L}_1)$. By Theorem \ref{trace thm} the derivative of $R_{\rel,\lambda}$ is integrable in $\lambda$. An application of Hille's theorem to the Bochner integral defining $g_\rel$ in the Banach space of trace-class operators shows that differentiation in $\epsilon$ commutes with integration. We therefore know that $g_\rel$ is differentiable and
$$\delta_Y g_\rel=\frac{\rmi}{2\pi}\int_{\Gamma}  g(z) \delta_Y R_{\rel}(z) \dd z=\frac{\rmi}{\pi}\int_{\tilde{\Gamma}} \lambda f(\lambda) \delta_Y R_{\rel,\lambda} \dd \lambda \quad \text{with} \quad R_{\rel}(\lambda^2)=R_{\rel,\lambda}.$$
Let $f(\lambda)=g(\lambda^2)$. 
Using Proposition \ref{trace variation prop} and integration by parts in $\lambda$, we have
\begin{align*}
\Tr(\delta_Y g_\rel)=&\frac{\rmi}{\pi}\int_{\tilde{\Gamma}} \lambda f(\lambda) \Tr(\delta_YR_{\rel,\lambda}) \dd \lambda
\\
=&\frac{\rmi}{2\pi}\int_{\tilde{\Gamma}} f(\lambda)\sum_{i=1}^N\Tr_{\partial \calO_i}\left[ \left(\partial_\nu\partial_\nu'\frac{\partial  G_{\calO,\lambda}}{\partial \lambda}-\partial_\nu\partial_\nu'\frac{\partial G_{\calO_i,\lambda}}{\partial \lambda}\right)\langle Y,\nu\rangle \right]_+ \dd \lambda
\\
=&-\sum_{i=1}^N\Tr_{\partial \calO_i}\left[ \partial_\nu\partial_\nu'(\breve g'_\calO-\breve g'_{\calO_i})\langle Y,\nu\rangle\right]_+.
\end{align*}
\end{proof}

In the special case $g(z) = \sqrt{z+m^2}$ Proposition \ref{trace of grel}  shows differentiability of  $H_\rel$ with respect to $\epsilon$ in the space of trace-class operators at $\epsilon=0$. Using Theorem \ref{Erel thm} and differentiating under the trace then gives the following theorem.

\begin{theorem}
	\label{variation of energy thm}	\label{variation of energy thm eqn main}
	The variation of the relative energy is given by $
	\delta_Y E_\rel=\frac{1}{2}\Tr(\delta_Y H_\rel).$
\end{theorem}

We will now use the Hadamard variation formula to compute this variation.

\subsection{Variation of the Klein-Gordon energy tensor}

\begin{theorem}
	\label{Equivalent thm}
Let $Y$ be a smooth boundary translation vector field. The variation of the Klein-Gordon energy generated by $Y$ is equal to the boundary integral of its spatial tensor contracted with $Y$. That is
\begin{equation*}
\label{Equivalent thm eqn main}
\delta_Y E_{\rel}=-\int_{\partial \calO} (T_\rel)_{ij}\nu^i Y^j \dd \sigma,
\end{equation*}
where the integration on the right-hand-side is at the exterior boundary and $\nu$ is the exterior normal for $\calE$.
\end{theorem}
\begin{proof}
From equations \eqref{stress energy tensor}, we have
\begin{equation*}
\int_{\partial \calO} (T_\rel)_{ij}\nu^i Y^j \dd \sigma
=\sum_{p=1}^{N}\int_{\partial \calO_p}(T_\rel)_{ij}\nu^i Y^j \dd \sigma
=\sum_{p=1}^{N}\int_{\partial \calO_p}((T_\ren)_\calO -\sum_{q=1}^N (T_\ren)_{\calO_q})_{ij}\nu^i Y^j \dd \sigma.
\end{equation*}
We know that
\begin{itemize}
	\item $((T_\ren)_{\calO_p})_{ij}$ is smooth on a neighbourhood of $\calO_q$ for $p\ne q$ (by Theorem \ref{DF of T}),
	\item $((T_\ren)_{\calO_p})_{ij}$ is divergence-free (by Theorem \ref{DF of T}),
	\item $Y$ is constant on a neighbourhood of $\calO_q$ (by assumptions) for all $q \in \{1, \ldots N\}$.
\end{itemize}
Therefore, $F_i=( (T_\ren)_{\calO_p})_{ij} Y^j$ is smooth and divergence-free on a neighbourhood of $\calO_q$. In other words, we have
\begin{equation*}
\int_{\partial \calO_q}( (T_\ren)_{\calO_p})_{ij}\nu^i Y^j \dd \sigma=0\quad\text{for}\quad q\ne p.
\end{equation*}
That is 
\begin{equation*}
\label{TrelTOTp}
\int_{\partial \calO} (T_\rel)_{ij}\nu^i Y^j \dd \sigma=\sum_{p=1}^{N}\int_{\partial \calO_p}((T_\ren)_\calO -(T_\ren)_{\calO_p})_{ij}\nu^i Y^j \dd \sigma.
\end{equation*}
As $((T_\ren)_\calO-(T_\ren)_{\calO_p})_{ij}$ is vanishing at the boundary $\partial \calO_p$, we have
\begin{equation*}
((T_\ren)_\calO -(T_\ren)_{\calO_p})_{ij}Y^j=((T_\ren)_\calO -(T_\ren)_{\calO_p})_{ij}\nu^j \langle Y,\nu\rangle \quad\text{on}\quad \partial \calO_p.
\end{equation*}
Now, from Theorem \ref{expression of T} we get
\begin{align*}
&\sum_{p=1}^{N}\int_{\partial \calO_p}((T_\ren)_\calO -(T_\ren)_{\calO_p})_{ij}\nu^i Y^j \dd \sigma
\\
=&\sum_{p=1}^{N}\int_{\partial \calO_p}((T_\ren)_\calO -(T_\ren)_{\calO_p})_{ij}\nu^i \nu^j \langle Y,\nu\rangle \dd \sigma
\\
=&\sum_{p=1}^{N}\int_{\partial \calO_p} \frac{1}{2}\left\{\left[[\partial_\nu\partial_\nu'(\breve H_\calO^{-1}-\breve H_\free^{-1})]|_\diag -\frac{1}{4}\Delta[( \breve H_\calO^{-1}-\breve H_\free^{-1})|_\diag]\right]\right.
\\
&\left.-\left[[\partial_\nu\partial_\nu'(\breve H_{\calO_p}^{-1}-\breve H_\free^{-1})]|_\diag -\frac{1}{4}\Delta[(\breve H_{\calO_p}^{-1}-\breve H_\free^{-1})|_\diag]\right]\right\}\langle Y,\nu\rangle \dd \sigma
\\
=&\frac{1}{2}\sum_{p=1}^{N}\int_{\partial \calO_p} \left[[\partial_\nu\partial_\nu'(\breve H_\calO^{-1}-\breve H_{\calO_p}^{-1})]|_\diag -\frac{1}{4}\Delta[(\breve H_\calO^{-1}-\breve H_{\calO_p}^{-1})|_\diag]\right]\langle Y,\nu\rangle \dd \sigma .
\end{align*}
Altogether, we have
\begin{align}
\label{Equivalent thm eqn 2}
&\int_{\partial \calO} (T_\rel)_{ij}\nu^i Y^j \dd \sigma \nonumber\\
&=\frac{1}{2}\sum_{p=1}^{N}\int_{\partial \calO_p} \left[[\partial_\nu\partial_\nu'(\breve H_\calO^{-1}-\breve H_{\calO_p}^{-1})]|_\diag -\frac{1}{4}\Delta[(\breve H_\calO^{-1}-\breve H_{\calO_p}^{-1})|_\diag]\right]\langle Y,\nu\rangle \dd \sigma .
\end{align}
The second term in the above equation can be expressed as
\begin{align*}
\Tr_{\partial \calO_p}\left[\Delta((\breve H_\calO^{-1}-\breve H_{\calO_p}^{-1})|_\diag)\right]
=&-\frac{\rmi}{2\pi}\int_{\Gamma} \frac{1}{\sqrt{z}} \Tr_{\partial \calO_p}\left[\partial_k\partial_k ((G_{\calO}(z)-G_{\calO_p}(z))|_\diag)\right] \dd z
\\
=&-\frac{\rmi}{2\pi}\int_{\Gamma} \frac{1}{\sqrt{z}} \Tr_{\partial \calO_p}\left[2\partial_k\partial_k'(G_{\calO}(z)-G_{\calO_p}(z))|_\diag\right] \dd z,
\end{align*}
where we used the properties \eqref{symmetric of T1} and \eqref{resolvent equation} of $G_{\calO}(\lambda^2)=G_{\calO,\lambda}$ and $G_{\calO_p}(\lambda^2)=G_{\calO_p,\lambda}$. Now we obtain
\begin{equation}
\label{Equivalent thm eqn 3}
\Tr_{\partial \calO_p}\left[\Delta((\breve H_\calO^{-1}-\breve H_{\calO_p}^{-1})|_\diag)\right]=2\Tr_{\partial \calO_p}\left[[\partial_\nu\partial_\nu'(\breve H_\calO^{-1}-\breve H_{\calO_p}^{-1})]|_\diag\right].
\end{equation}
Equations \eqref{Equivalent thm eqn 2} and \eqref{Equivalent thm eqn 3} imply
\begin{equation*}
\int_{\partial \calO} (T_\rel)_{ij}\nu^i Y^j \dd \sigma
=\frac{1}{4}\sum_{p=1}^{N}\int_{\partial \calO_p} \left[[\partial_\nu\partial_\nu'(\breve H_\calO^{-1}-\breve H_{\calO_p}^{-1})]|_\diag\right]\langle Y,\nu\rangle \dd \sigma .
\end{equation*}
Since $\int_{\partial \calO} (T_\rel)_{ij}\nu^i Y^j \dd \sigma$ is integrating at the exterior boundary, we have 
\begin{equation}
\label{Equivalent thm eqn 4}
\int_{\partial \calO} (T_\rel)_{ij}\nu^i Y^j \dd \sigma
=\frac{1}{4}\sum_{p=1}^{N}\int_{\partial \calO_p} \left[[\partial_\nu\partial_\nu'(\breve H_\calO^{-1}-\breve H_{\calO_p}^{-1})]|_\diag\right]_+\langle Y,\nu\rangle \dd \sigma .
\end{equation}
Applying Proposition \ref{trace of grel} to $g(z)=\sqrt{z+m^2}$, we have
\begin{align*}
\Tr(\delta_Y H_\rel)=&\Tr\left(\delta_Y \left(H_\calO-\sum_{p=1}^N H_{\calO_p}\right)\right)
\\
=&(\Tr_\calE+\Tr_\calO)\left[\delta_Y \left(H_\calO-\sum_{p=1}^N H_{\calO_p}\right)\right]
\\
=&-\frac{1}{2}\left\{\Tr_{\partial \calO}[ \partial_\nu\partial_\nu' \breve H_\calO^{-1}\langle Y,\nu\rangle]_+-\sum_{p=1}^N\Tr_{\partial \calO_p}[ \partial_\nu\partial_\nu' \breve H_{\calO_p}^{-1}\langle Y,\nu\rangle]_+\right\}
\\
&-\frac{1}{2}\left\{\Tr_{\partial \calO}[ \partial_\nu\partial_\nu' \breve H_\calO^{-1}\langle Y,\nu\rangle]_--\sum_{p=1}^N\Tr_{\partial \calO_p}[ \partial_\nu\partial_\nu' \breve H_{\calO_p}^{-1}\langle Y,\nu\rangle]_-\right\} .
\end{align*}
Now $\Tr_{\partial \calO}[ \partial_\nu\partial_\nu' \breve H_\calO^{-1}\langle Y,\nu\rangle]_--\sum_{p=1}^N\Tr_{\partial \calO_p}[ \partial_\nu\partial_\nu' \breve H_{\calO_p}^{-1}\langle Y,\nu\rangle]_-=0$ as $Y$ is a boundary translation vector field. Therefore, we have
\begin{align*}
\Tr(\delta_Y H_\rel)
=&-\frac{1}{2}\left\{\Tr_{\partial \calO}[ \partial_\nu\partial_\nu'\breve H_\calO^{-1}\langle Y,\nu\rangle]_+-\sum_{p=1}^N\Tr_{\partial \calO_p}[ \partial_\nu\partial_\nu' \breve H_{\calO_p}^{-1}\langle Y,\nu\rangle]_+\right\}
\\
=&-\frac{1}{2}\sum_{p=1}^N\left\{\Tr_{\partial \calO_p}[ \partial_\nu\partial_\nu' \breve H_\calO^{-1}\langle Y,\nu\rangle]_+-\Tr_{\partial \calO_p}[ \partial_\nu\partial_\nu' \breve H_{\calO_p}^{-1}\langle Y,\nu\rangle]_+\right\}
\\
=&-\frac{1}{2}\sum_{p=1}^N\left\{\Tr_{\partial \calO_p}[ \partial_\nu\partial_\nu'(\breve H_\calO^{-1}- \breve H_{\calO_p}^{-1})\langle Y,\nu\rangle]_+\right\}
\\
=&-\frac{1}{2}\sum_{p=1}^{N}\int_{\partial \calO_p} \left[[\partial_\nu\partial_\nu'( \breve H_\calO^{-1}-\breve H_{\calO_p}^{-1})]|_\diag\right]_+ \langle Y,\nu\rangle \dd \sigma .
\end{align*}
That is 
\begin{equation}
\label{Tr delta Hrel}
\Tr(\delta_Y H_\rel)=-\frac{1}{2}\sum_{p=1}^{N}\int_{\partial \calO_p} \left[[\partial_\nu\partial_\nu'(\breve H_\calO^{-1}-\breve H_{\calO_p}^{-1})]|_\diag\right]_+ \langle Y,\nu\rangle \dd \sigma .
\end{equation}
Finally, equations \eqref{Equivalent thm eqn 4}, \eqref{Tr delta Hrel} and Theorem \ref{variation of energy thm} complete the proof.
\end{proof}
An application of an analogue of Theorem \ref{Equivalent thm} to calculate the Casimir force in dimension one can be found in Appendix \ref{1-D relative setting}.
Finally, we have the following theorem.
\begin{theorem}
	\label{Equivalent thm 1}
Let $T_\ren$ be the renormalised stress energy tensor in Theorem \ref{def of Tren} and let $Y$ be a boundary translation flow as in Theorem \ref{Equivalent thm}. 
We assume further that $Y$ is constant near $\calO_p$ for some $p \in \{1,\dots,N\}$ and vanishes near $\calO_q$ if $q \not= p$. Let
 $S$ be any smooth hypersurface in $\calE$ that is homologous to $\partial \calO_p$ in $\overline{\calE}$ and let $\nu$ be the exterior normal vector field of $S$. Then the variation of the relative energy generated by $Y$ is equal to 
 \begin{equation*}
\delta_Y E_{\rel}=\int_{S} (T_\ren)_{ij}\nu^i Z^j \dd \sigma,
\end{equation*}
where $Z$ is the unique constant vector field on $\R^d$ whose restriction to $\partial \calO_q$ equals $Y$.
\end{theorem}
\begin{proof}
As in the proof of Theorem \ref{Equivalent thm}, we have
\begin{align*}
\int_{\partial \calO} (T_\rel)_{ij}\nu^i Y^j \dd \sigma=&\int_{\partial \calO_p}((T_\ren)_\calO -(T_\ren)_{\calO_p})_{ij}\nu^i Z^j \dd \sigma
\\
=&\int_{\partial \calO_p}((T_\ren)_\calO -(T_\ren)_{\calO_p})_{ij}\nu^i \nu^j \langle Z,\nu\rangle \dd \sigma .
\end{align*}
 We also know that $\partial \calO_p$ is homologous to $S_{R}$, a sphere with sufficiently large radius $R$ in $\R^d\backslash \calO_p$. Because $(T_\ren)_{\calO_p}$ is divergent-free, one has
\begin{equation*}
\int_{\partial \calO_p}((T_\ren)_{\calO_p})_{ij}\nu^i \nu^j \langle Z,\nu\rangle \dd \sigma = \lim_{R\to \infty} \int_{S_R}((T_\ren)_{\calO_p})_{ij}\nu^i \nu^j \langle Z,\nu\rangle \dd \sigma.
\end{equation*}
To get a decay property of $(T_\ren)_{\calO_p}$ at infinity, we first recall that
\begin{equation*}
H_{\calO} - H_\free = -\frac{2}{\pi}\int_m^\infty {\lambda}{\sqrt{\lambda^2-m^2}} \left( R_{\calO,\rmi \lambda}-R_{\free,\rmi \lambda} \right) \der \lambda.
\end{equation*}	
For $m\ne 0$, we would have an exponential decay of $e^{-Cm \rho(x)}$ for $\left(\breve H_{\calO} - \breve H_\free\right)(x,x)$, as explained in the proof of Theorem \ref{integrability thm}. For $m=0$, one could use the estimates of \eqref{d3kernel} and \eqref{d2kernel} to obtain that, for $d\ge 2$,
$$
|\left(\breve H_{\calO} - \breve H_\free\right)(x,x)|\le C \rho^{-2d+1}(x).
$$
Moreover, estimate \eqref{DofHrel} also implies
$$
|\left(\Delta(\breve  H_\calO^{-1}-\breve H_\free^{-1})\right)(x,x)|\le C \rho^{-2d+1}(x).
$$
Applying the above estimates to Theorem \ref{expression of T}, one concludes that 
$$|(T_\ren)_{\calO_p}|\le C \rho^{-2d+1}(x).$$
This implies $\lim_{R\to \infty} \int_{S_R}((T_\ren)_{\calO_p})_{ij}\nu^i \nu^j \langle Y,\nu\rangle \dd \sigma =0$.

As $(T_\ren)_\calO=T_\ren$ (see Definition \ref{def of Tren} and the first two paragraphs of Section \ref{The relative trace-formula and the Casimir energy}), this completes the proof.
\end{proof}

\section{The zeta regularised energy and the equivalence of (1) and (3)} \label{zetaregen}

In this section we assume that $m>0$. In that case it is well known that for $\Re(s)>(d+1)/2$ the operator
$$
 \left( H^{-2s}_\calO - H^{-2s}_\free \right)
$$
is trace-class (see for example \cite{borisov1988relative}) and we can therefore define the renormalised zeta functions $\zeta_\calO(s)$ as
$$
 \zeta_\calO(s)=  \tr \left( H^{-2s}_\calO - H^{-2s}_\free \right) = -2 s \int_0^\infty \lambda (\lambda^2 + m^2)^{-s-1} \xi(\lambda) \der \lambda
$$
where $\xi(\lambda)$ is the spectral shift function of the problem. 
The Birman-Krein formula applies to this setting and we have
$$
 \xi(\lambda) = -\frac{1}{2 \pi \rmi}\log \det( S(\lambda))
$$
where $S(\lambda)$ is the stationary scattering matrix of the problem. 

One can also use the Mellin tranform to write
$$
 \zeta_\calO(s)= \frac{1}{\Gamma(s)} \int_0^\infty t^{s-1}  h_\calO(t) \der t,
 $$
 where
 \begin{align*}
  h_\calO(t) =  \tr \left( e^{-t H_\calO^2} - e^{-t H^2_{\free}}  \right) &=e^{-m^2 t} \tr \left( e^{t \Delta_\calO} - e^{t \Delta_{\free}}  \right)  \\&=  -2 t e^{-m^2 t} \int_0^\infty e^{-\lambda^2 t} \xi_\calO(\lambda) \lambda \der \lambda.
 \end{align*}
 The following Lemma should be well known but we could not find a reference for the precise statement. It is a simple consequence of heat kernel expansions \cite{BransonGilkey1990,Greiner1971,GilkeySmith1983} and Kac's principle of not feeling the boundary \cite{Kac1951,MR3584189}. We also refer to \cite{JensenKato,Kato1978} for more details on obstacle scattering theory and the Birman-Krein formula.
 \begin{lemma}
 The function $h_\calO(t)$ is exponentially decaying as $t \to \infty$ and has a full asymptotic expansion as $t \searrow 0$ of the form
 $$
  h_\calO(t) \sim t^{-(d-1)/2} \sum_{k=0}^\infty a_k t^{k}, \quad t \to 0_+,
 $$
 where the infinite sum is understood in the sense of asymptotic summation.
 The coefficients $a_k$ are integrals over $\partial\calO$ of locally determined quantities expressed in terms of the extrinsic and intrinsic curvature of the boundary and its derivatives. In particular,
 $$
  a_0 = -\frac{1}{4}(4 \pi )^{-\frac{d-1}{2}} \mathrm{Vol}_{d-1}(\partial \calO).
  $$
 \end{lemma}
\begin{proof}
The exponential decay follows immediately from the representation by means of the spectral shift function and $m>0$. Now  $h_\calO(t)$
is the trace of the difference of the two heat operators $e^{-t H^2_\calO}$ and $e^{-t H^2_{\free}}$ with integral kernels  $K_\calO(t,x,y)$ and $K_\free(t,x,y)$, respectively. Since the difference is trace-class and the integral kernel is smooth we have
$$
h_\calO(t) = \lim_{R \to \infty} \int_{\mathcal{E} \cap B_R }K_\calO(t,x,x) - K_\free(t,x,x) \der x,
$$
where $B_R$ is a ball of radius $R$, i.e. integration is over a large ball of radius $R$ with the obstacles removed.
The heat kernel difference satisfies not feeling the boundary estimates. For example a general finite propagation speed estimate (\cite{MR3584189}) gives
$$
 |  K_\calO(t,x,x) - K_\free(t,x,x) | \leq C \left( \rho(x)^{-d} +1 \right) t^{-d-\frac{3}{2}} e^{-\frac{\rho(x)^2}{ t}},
$$
where $\rho(x) = \mathrm{dist}(x,\calO)$.
This shows that
$$
 h_\calO(t) = \int_{\mathcal{E}}\left( K_\calO(t,x,x) - K_\free(t,x,x) \right)\der x.
$$
Let $\mathcal{U}\subset \mathcal{E}$ be an open neighbourhood that contains $\partial \mathcal{E}$, i.e. $\partial \mathcal{E}\subset\mathcal{U}\subset \mathcal{E}$. Then we have for $|t|<1$ that
$$
 h_\calO(t) = \int_{\mathcal{U}} \left( K_\calO(t,x,x) - K_\free(t,x,x) \right) \der x + O(t^N),
$$
for any $N>0$.
This computation is therefore purely local. Using another not feeling the boundary estimate we can replace $K_\calO(t,x,x)$ in this integral
by the Dirichlet heat kernel on the compact manifold obtained by removing $\calO$ from the large flat torus. The coefficients are therefore the same as the heat kernel coefficients on a domain with boundary. It is well-known that heat kernel coefficients are determined by local invariants of the jets of the symbols of the operators, i.e. jets of Riemannian metric and second fundamental forms (\cite[Lemma 2.6]{GilkeySmith1983} or \cite[Lemma 2.1]{BransonGilkey1990}). As our interior geometry is Euclidean, only the first heat coefficients corresponding to the interior is non-zero. This first coefficient is the same for both operators and therefore only boundary terms contribute to the expansion. The first non-trivial term is given by $a_0 = -\frac{1}{4}(4 \pi )^{-\frac{d-1}{2}} \mathrm{Vol}_{d-1}(\partial \calO)$ (see for example \cite[Theorem 1.1]{BransonGilkey1990}).
\end{proof}

\begin{rem}
 The general form of the heat expansion for a compact Riemannian manifold $M$ with boundary $\partial M$ is of the form
 $$
  \Tr e^{t \Delta_M}\sim t^{-\frac{d}{2}} \sum_{k=0}^\infty b_k t^{\frac{k}{2}}, \quad t \to 0_+,
  $$
  where the $b_{2k}$ are integrals of locally defined quantities over $M$, and the $b_{2k+1}$ are integrals of locally defined quantities over $\partial M$ which are determined by the boundary conditions.
  When considering differences heat kernels of Laplace operators with different boundary conditions the $b_{2k}$ terms cancel and only the terms containing $b_{2k+1}$ remain.
\end{rem}

It follows as usual (for example \cite[Section 1.12]{Gilkey}) that $\zeta_\calO$ has a meromorphic continuation to the complex plane. If $d$ is odd then there are finitely many poles
at $\{\frac{d-1}{2},\frac{d-1}{2}-1,\ldots, 1\}$ with residue at $\frac{d-1}{2}-k$ determined by the coefficients $a_k$. In this case the values at non-positive integers are also expressible in terms of $a_k$. In case $d$ is even poles may be located at the points
$\{\frac{d-1}{2}-k \;\mid \; k \in \mathbb{N}_0 \}$. 
\begin{definition} \label{regenergy}
The regularised energy $E_\reg$ is then defined 
$$
 E_\reg =   \frac{1}{2}  \mathrm{FP}_{s=-\frac{1}{2}}\left( \zeta_\calO(s) \right),
$$
where $ \mathrm{FP}_{z=a} f(z)$ denotes the finite part of the meromorphic function $f$ at the point $a$, i.e.
the constant term in the Laurent expansion of $f$ about the point $a$.
\end{definition}
In particular, in case $d$ is odd we have $E_\reg = \frac{1}{2} \zeta_\calO(-\frac{1}{2})$.

We can also define a zeta regularized energy $E^j_\reg$ for every object $\calO_j$. Obviously, $E^j_\reg$ does not depend on the position of $\calO_j$ in $\R^d$ and is also invariant under active rotations of the object. Since the heat coefficients are local quantities 
the relative zeta function
$$
  \zeta_\rel(s) = \zeta_\calO(s) - \sum_{j=1}^N \zeta_{\calO_j}(s)
$$
is an entire function. Since the relative quantities
$$
  \left( H^{-2s}_\calO - H^{-2s}_\free \right) - \sum_{j=1}^N \left( H^{-2s}_{\calO_j} - H^{-2s}_\free \right) 
$$
are trace-class for $s<0$ 
we also have that $E_\rel = \frac{1}{2} \zeta_\rel(-\frac{1}{2})$ and therefore
$$
 E_\rel = E_\reg - \sum_{j=1}^N E^j_\reg.
$$
Thus, $E_\rel - E_\reg$ does not change if the individual objects are translated or rotated.

\section{Proof of main theorem}
In this section, we will prove our main theorem (Theorem \ref{main thm}) by combining the results we obtained in the previous sections.
\begin{proof}[Proof of Theorem \ref{main thm}]
The differentiability of $E_\rel(\epsilon)$ follows from Theorem \ref{existence thm} and \ref{trace thm}.  As shown in Section \ref{The relative trace-formula and the Casimir energy}, we have $\Xi(\lambda) = \log \det(Q_\lambda \tilde Q_\lambda^{-1})$ and Theorem \ref{Erel thm}. Recall that equation \eqref{Erel eqn}, says
\begin{equation*}
E_\rel = \frac{1}{2} \Tr\, H_{\rel} = \frac{1}{2 \pi} \int_{m}^\infty \frac{\omega}{ \sqrt{\omega^2 - m^2}} \Xi(\rmi \omega) \der \omega.
\end{equation*}
By substituting $\Xi(\lambda)$, one obtains equation \eqref{main thm eqn2} of our main theorem. Equation \eqref{main thm eqn1} follows immediately from Theorem \ref{Equivalent thm} and \ref{Equivalent thm 1}. Since the flow $\Phi_\epsilon$ described right above Theorem \ref{main thm} is exactly a boundary translation flow (see Definition \ref{def of BT flow}), we know from the end of Section \ref{zetaregen} that $E^j_\reg$ is constant if the individual objects are translated. Hence, the fact that $E_\rel(\epsilon) - E_{\reg}(\epsilon)$ is constant near $\epsilon=0$ for $m>0$ follows.
\end{proof}

\appendix

\section{Bounds on the free resolvent}
\label{Bounds on the free resolvent}
In this appendix, we will give some estimates on the kernels of $\phi G_{\free,\lambda}\phi$ and $\phi G_{\free,\lambda}\chi$, which are denoted by $K_\lambda(x,y)$ and $\tilde{K}_\lambda(x,y)$ respectively. They are given by
\begin{equation*}
\left\{
\begin{aligned}
K_\lambda(x,y)=\frac{\rmi}{4}\phi(x)\phi(y)\left(\frac{\lambda}{2\pi |x-y|}\right)^{\frac{d-2}{2}}H^{(1)}_{\frac{d-2}{2}}(\lambda|x-y|),\\
\tilde{K}_\lambda(x,y)=\frac{\rmi}{4}\phi(x)\chi(y)\left(\frac{\lambda}{2\pi |x-y|}\right)^{\frac{d-2}{2}}H^{(1)}_{\frac{d-2}{2}}(\lambda|x-y|),
\end{aligned}
\right.
\end{equation*}
where $H^{(1)}_\nu$ is the Hankel function. Moreover, we assume $\phi$ has compact support whereas the support of $\chi$ is unbounded with $\dist(\supp \phi,\supp \chi)=\delta>0$. Now we have
\begin{align*}
\calK_{\mu,\nu,p,d}(x;\lambda)&=\int_{\R^d}\left|\phi(x)\phi(y)\left(\frac{\lambda}{2\pi |x-y|}\right)^{\mu}H^{(1)}_{\nu}(\lambda|x-y|)\right|^p\dd y
\\
&\le C \int_{\supp(\phi)}\left|\left(\frac{\lambda}{2\pi |x-y|}\right)^{\mu}H^{(1)}_{\nu}(\lambda|x-y|) \right|^p \dd y
\\
&\le C \int_{\mathbb{S}^{d-1}}\int_0^{r_\phi}\left|\left(\frac{\lambda}{2\pi r}\right)^{\mu}H^{(1)}_{\nu}(\lambda r) \right|^p r^{d-1} \dd r \dd \omega
\\
&\le C \int_0^{r_\phi}\left|\left(\frac{\lambda}{r}\right)^{\mu}H^{(1)}_{\nu}(\lambda r) \right|^p r^{d-1} \dd r .
\end{align*}
Similarly,
$$
\tilde{\calK}_{\mu,\nu,p,d}(x;\lambda)\le C \int_\delta^{\infty}\left|\left(\frac{\lambda}{r}\right)^{\mu}H^{(1)}_{\nu}(\lambda r) \right|^p r^{d-1} \dd r .
$$
Recall that Hankel functions have the following asymptotic
\begin{align}
&\text{For $|z|<r_0$}, & &|H^{(1)}_\nu(z)|\le C_{r_0,\nu} \left\{
\begin{aligned}
&|z|^{-\nu}&\text{if $\nu>0$},\\
&|\log (z)|&\text{if $\nu=0$}.
\end{aligned}
\right.
\label{Hankelsmall}\\
&\text{For $|z|\ge r_0$}, & &|H^{(1)}_\nu(z)|\le C_{r_0,\nu} |z|^{-\frac{1}{2}}e^{-\Im z} .\label{Hankellarge}
\end{align}
Therefore, for small $\lambda$ such that $r_\phi|\lambda| < r_0$, we have
\begin{equation*}
\int_0^{r_\phi}\left|\left(\frac{\lambda}{r}\right)^{\mu}H^{(1)}_{\nu}(\lambda r) \right|^p r^{d-1} \dd r 
\le \left\{
\begin{aligned}
&C \left|\lambda^\mu \log(\lambda)\right|^p \int_0^{r_\phi} r^{d-\mu p-1}|\log r|^p\dd r & &\text{for $\nu=0$} ,\\
&C |\lambda|^{(\mu-\nu)p}\int_0^{r_\phi}r^{d-1-(\mu+\nu)p} \dd r & &\text{for $\nu \ne 0$} .
\end{aligned}
\right.
\end{equation*}
For large $\lambda$ such that $r_\phi|\lambda| \ge r_0$, we have
\begin{align*}
&\int_0^{r_\phi}\left|\left(\frac{\lambda}{r}\right)^{\mu}H^{(1)}_{\nu}(\lambda r) \right|^p r^{d-1} \dd r  
\\
=&|\lambda|^{2\mu p-d}\int_0^{r_\phi|\lambda|}\left|H^{(1)}_{\nu}\left(\frac{\lambda}{|\lambda|}r\right) \right|^p r^{d-\mu p-1} \dd r  
\\
=&|\lambda|^{2\mu p-d}\left(\int_0^{r_0}\left|H^{(1)}_{\nu}\left(\frac{\lambda}{|\lambda|}r\right) \right|^p r^{d-\mu p-1} \dd r+\int_{r_0}^{r_\phi|\lambda|}\left|H^{(1)}_{\nu}\left(\frac{\lambda}{|\lambda|}r\right) \right|^p r^{d-\mu p-1} \dd r\right) .
\end{align*}
The above terms can be bounded by using the asymptotic of Hankel functions in equations \eqref{Hankelsmall} and  \eqref{Hankellarge}.
\begin{equation*}
\left\{
\begin{aligned}
&|\lambda|^{2\mu p-d}\left(\int_0^{r_0}\left|\log\left(\frac{\lambda}{|\lambda|}r\right) \right|^p r^{d-\mu p-1} \dd r  +\int_{r_0}^{r_\phi|\lambda|} r^{-\frac{1}{2}}e^{-r\frac{\Im \lambda}{|\lambda|}} r^{d-\mu p-1}\dd r \right)& &\text{for $\nu=0$},
\\
&|\lambda|^{2\mu p-d}\left(\int_0^{r_0}r^{d-(\mu+\nu) p-1} \dd r +\int_{r_0}^{r_\phi|\lambda|} r^{-\frac{1}{2}}e^{-r\frac{\Im \lambda}{|\lambda|}} r^{d-\mu p-1}\dd r \right) & &\text{for $\nu\ne 0$}.
\end{aligned}
\right.
\end{equation*}
Replacing $r_\phi|\lambda|$ by $\infty$ in the integral, we obtain
$$
\calK_{\mu,\nu,p,d}(x;\lambda)  \le C_{r_0,\mu,\nu,p,d}|\lambda|^{2\mu p-d} \quad \text{for large $|\lambda|$}.
$$
Combining with the asymptotic for small $|\lambda|$, we have
\begin{equation}
\label{Kmunupd}
\calK_{\mu,\nu,p,d}(x;\lambda)
\le \left\{
\begin{aligned}
&C_{r_0,\mu,p,d}\, \rho_{p,\mu p;2\mu p -d}(\Im \lambda) & &\text{for $\nu=0$} ,\\
&C_{r_0,\mu,\nu,p,d}\, \rho_{0,(\mu-\nu) p;2\mu p -d}(\Im \lambda) & &\text{for $\nu \ne 0$} .
\end{aligned}
\right.
\end{equation}

For $\tilde{\calK}_{\mu,\nu,p,d}$, we have
\begin{equation*}
\tilde{\calK}_{\mu,\nu,p,d}(x;\lambda)\le C \int_\delta^{\infty}\left|\left(\frac{\lambda}{r}\right)^{\mu}H^{(1)}_{\nu}(\lambda r) \right|^p r^{d-1} \dd r .
\end{equation*}
For small $\lambda$ such that $\delta |\lambda| < r_0$, we also have
\begin{align*}
&\int_\delta^{\infty}\left|\left(\frac{\lambda}{r}\right)^{\mu}H^{(1)}_{\nu}(\lambda r) \right|^p r^{d-1} \dd r 
\\
=&|\lambda|^{2\mu p-d}\int_{\delta|\lambda|}^{\infty}\left|H^{(1)}_{\nu}\left(\frac{\lambda}{|\lambda|}r\right) \right|^p r^{d-\mu p-1} \dd r  
\\
=&|\lambda|^{2\mu p-d}\left(\int_{\delta|\lambda|}^{r_0} \left|H^{(1)}_{\nu}\left(\frac{\lambda}{|\lambda|}r\right) \right|^p r^{d-\mu p-1} \dd r + \int_{r_0}^{\infty}\left|H^{(1)}_{\nu}\left(\frac{\lambda}{|\lambda|}r\right) \right|^p r^{d-\mu p-1} \dd r \right).
\end{align*}
This can be bounded by
\begin{equation*}
\left\{
\begin{aligned}
&|\lambda|^{2\mu p-d}\left(\int_{\delta|\lambda|}^{r_0}\left|\log\left(\frac{\lambda}{|\lambda|}r\right) \right|^p r^{d-\mu p-1} \dd r  +\int_{r_0}^{\infty} r^{-\frac{1}{2}}e^{-r\frac{\Im \lambda}{|\lambda|}} r^{d-\mu p-1}\dd r \right)& &\text{for $\nu=0$},
\\
&|\lambda|^{2\mu p-d}\left(\int_{\delta|\lambda|}^{r_0}r^{d-(\mu+\nu) p-1} \dd r +\int_{r_0}^{\infty} r^{-\frac{1}{2}}e^{-r\frac{\Im \lambda}{|\lambda|}} r^{d-\mu p-1}\dd r \right) & &\text{for $\nu\ne 0$}.
\end{aligned}
\right.
\end{equation*}
Replacing $\delta|\lambda|$ by $0$, one has
$$
\tilde{\calK}_{\mu,\nu,p}(x;\lambda)\le C_{\mu,\nu,p,d}|\lambda|^{2\mu p-d} \quad \text{for small $|\lambda|$}.
$$
For large $\lambda$ such that $\delta |\lambda| \ge r_0$, we have
\begin{align*}
\int_\delta^{\infty}\left|\left(\frac{\lambda}{r}\right)^{\mu}H^{(1)}_{\nu}(\lambda r) \right|^p r^{d-1} \dd r 
\le& C \int_\delta^{\infty}\left|\left(\frac{\lambda}{r}\right)^{\mu}|\lambda r|^{-\frac{1}{2}}e^{-r\Im \lambda} \right|^p r^{d-1} \dd r 
\\
\le& C |\lambda|^{\mu p-\frac{p}{2}}  \int_\delta^{\infty} r^{d-\mu p-\frac{p}{2}-1}e^{-r\Im \lambda} \dd r 
\\
\le& C |\lambda|^{\mu p-\frac{p}{2}}  e^{-\delta \Im \lambda}\int_0^{\infty} (r+\delta)^{d-\mu p-\frac{p}{2}-1}e^{-r\Im \lambda} \dd r 
\\
\le& C_{\delta,r_0,\mu,p,d}\, e^{-\delta \Im \lambda}.
\end{align*}
This implies
\begin{equation}
\label{tildeKmunupd}
\tilde{\calK}_{\mu,\nu,p,d}(x;\lambda) \le C_{r_0,\mu,\nu,p,d}\, \rho_{0,2\mu p -d;0}(\Im \lambda)\, e^{-\delta \Im \lambda}.
\end{equation}
By the Schur test and equation \eqref{Kmunupd}, we have 
\begin{equation*}
\|\phi G_{\free,\lambda}\phi\|_{L^2\to L^2}\le \sup_{x\in\supp(\phi)}\calK_{\frac{d-2}{2},\frac{d-2}{2},1,d} (x;\lambda)\le C_{r_0,\phi}
\left\{
\begin{aligned}
&\rho_{1,0;-2}(\Im \lambda) & &\text{for $d=2$},\\
&\rho_{0,0;-2}(\Im \lambda) & &\text{for $d\ne 2$}.
\end{aligned}
\right.
\end{equation*}
Using the recurrence relations of Hankel functions $(H^{(1)}_{\nu})'(z)=-H^{(1)}_{\nu+1}(z)+\frac{\nu}{z}H^{(1)}_{\nu}(z)$ and similar calculations, one has
\begin{equation*}
\|\nabla_i \circ (\phi  G_{\free,\lambda}\phi)\|_{L^2\to L^2}\le C_{r_0,\phi}
\left\{
\begin{aligned}
&\rho_{1,0;-1}(\Im \lambda) & &\text{for $d=2$},\\
&\rho_{0,0;-1}(\Im \lambda) & &\text{for $d\ne 2$}.
\end{aligned}
\right.
\end{equation*}
This is the same bound for $\|\phi G_{\free,\lambda}\phi\|_{L^2\to H^1}$. To get the bound for $\|\phi G_{\free,\lambda}\phi\|_{L^2\to H^2}$, we recall that  $(-\Delta-\lambda^2)G_{\free,\lambda}=I$, hence
\begin{align*}
\|\phi G_{\free,\lambda}\phi\|_{L^2\to H^2}&\le C \|\Delta \phi G_{\free,\lambda}\phi\|_{L^2\to L^2}+C\|\phi G_{\free,\lambda}\phi\|_{L^2\to L2}\\
&\le C(1+|\lambda|^2)\|\phi G_{\free,\lambda}\phi\|_{L^2\to L^2}+C\|\tilde{\phi} G_{\free,\lambda}\tilde{\phi}\|_{L^2\to H^1},
\end{align*}
where $\tilde{\phi}=1$ near $\supp \phi$ and $\operatorname{diam}(\supp \tilde{\phi})$ is bounded. Combining with the estimates on $\|\phi G_{\free,\lambda}\phi\|_{L^2\to H^1}$ and $\|\phi G_{\free,\lambda}\phi\|_{L^2\to L^2}$, we have
\begin{equation*}
\|\phi G_{\free,\lambda}\phi\|_{L^2\to H^2}\le C_{r_0,\phi}
\left\{
\begin{aligned}
&\rho_{1,0;0}(\Im \lambda) & &\text{for $d=2$},\\
&1 & &\text{for $d\ne 2$}.
\end{aligned}
\right.
\end{equation*}
For the operator $\chi G_{\free,\lambda}\phi$, we obtain from equation \eqref{tildeKmunupd} that
\begin{align*}
\|\chi G_{\free,\lambda}\phi\|_{HS(H^{s}(\partial \calO)\to L^2(\R^d))}^2 &\le \sup_{x\in\supp(\phi)}\tilde{\calK}_{\frac{d-2}{2},\frac{d-2}{2},2,d} (x;\lambda) \\ &\le C_{r_0,\phi,\chi}\, \rho_{0,d-4;0}(\Im \lambda)e^{-\delta \Im \lambda},
\end{align*}
which gives the verification of the estimate \eqref{S prop 2 eqn 1}.

\section{The method illustrated for the 1-D Casimir effect}
\label{1-D relative setting}
In this appendix, we illustrate Theorem \ref{Equivalent thm} in its simplest form i.e. for the case of the $1$-dimensional Casimir effect with $m=0$. This will also illustrate the advantages of the relative framework.
Let $a_1<b_1<a_2<b_2$, where $\calO_1=(a_1,b_1)$ and $\calO_2=(a_2,b_2)$ are the obstacles. Then we have
\begin{equation*}
G_\calO=
\left\{
\begin{aligned}
&G_{(-\infty,a_1)}& x,y<a_1\\
&G_{(a_1,b_1)}& a_1<x,y<b_1\\
&G_{(b_1,a_2)}& b_1<x,y<a_2\\
&G_{(a_2,b_2)}& a_2<x,y<b_2\\
&G_{(b_2,+\infty)}& x,y>b_2
\end{aligned}
\right. ,
\end{equation*}

\begin{equation*}
G_{\calO_1}=
\left\{
\begin{aligned}
&G_{(-\infty,a_1)}& x,y<a_1\\
&G_{(a_1,b_1)}& a_1<x,y<b_1\\
&G_{(b_1,+\infty)}& x,y>b_1
\end{aligned}
\right. ,
\end{equation*}

\begin{equation*}
G_{\calO_2}=
\left\{
\begin{aligned}
&G_{(-\infty,a_2)}& x,y<a_2\\
&G_{(a_2,b_2)}& a_2<x,y<b_2\\
&G_{(b_2,+\infty)}& x,y>b_2
\end{aligned}
\right. .
\end{equation*}

Then $G_\rel=G_\calO-G_{\calO_1}-G_{\calO_2}+G_\free$ is given by

\begin{equation}
\label{1D Grel}
G_\rel=
\left\{
\begin{aligned}
&G_\free-G_{(-\infty,a_2)}& x,y<b_1\\
&G_{(b_1,a_2)}-G_{(b_1,+\infty)}-G_{(-\infty,a_2)}+G_\free& b_1<x,y<a_2\\
&G_\free-G_{(b_1,+\infty)}& x,y>a_2
\end{aligned}
\right. .
\end{equation}
In particular, 
\begin{equation*}
G_\free(x,y;k^2)=-\frac{e^{\rmi k|x-y|}}{2\rmi k}=
\left\{
\begin{aligned}
&-\frac{e^{-\rmi k(x-y)}}{2\rmi k}& x<y\\
&-\frac{e^{-\rmi k(y-x)}}{2\rmi k}& x>y
\end{aligned}
\right. ,
\end{equation*}

\begin{equation*}
G_{(b,+\infty)}(x,y;k^2)=-\frac{e^{\rmi k|x-y|}-e^{\rmi k|x+y-2b|}}{2\rmi k}=
\left\{
\begin{aligned}
&\frac{e^{\rmi k(y-b)}\sin(k(x-b))}{k}& x<y\\
&\frac{e^{\rmi k(x-b)}\sin(k(y-b))}{k}& x>y
\end{aligned}
\right. ,
\end{equation*}

\begin{equation*}
G_{(-\infty,a)}(x,y;k^2)=-\frac{e^{\rmi k|x-y|}-e^{\rmi k|x+y-2a|}}{2\rmi k}=
\left\{
\begin{aligned}
&\frac{e^{-\rmi k(x-a)}\sin(k(a-y))}{k}& x<y\\
&\frac{e^{-\rmi k(y-a)}\sin(k(a-x))}{k}& x>y
\end{aligned}
\right. ,
\end{equation*}

\begin{multline*}
G_{(a,b)}(x,y;k^2)=\frac{\cos(k(x+y-b-a))-\cos(k(b-a-|x-y|))}{2k\sin (k(b-a))}
\\
=
\left\{
\begin{aligned}
&\frac{\sin(k(x-a))\sin(k(b-y))}{k\sin(k(b-a))}& x<y\\
&\frac{\sin(k(y-a))\sin(k(b-x))}{k\sin(k(b-a))}& x>y
\end{aligned}
\right. .
\end{multline*}

For $x,y>b$, we have
\begin{equation*}
[G_\free-G_{(b,+\infty)}](x,y;k^2)=-\frac{e^{\rmi k|x+y-2b|}}{2\rmi k} ,
\end{equation*}
which implies
\begin{align*}
[\breve{H}_\free-\breve{H}_{(b,+\infty)}](x,y)
=&\frac{\rmi}{\pi}\int_{\tilde{\Gamma}}k \sqrt{k^2}[G_\free-G_{(b,+\infty)}](x,y;k^2) \dd k
\\
=&\frac{\rmi}{\pi}\int_{\tilde{\Gamma}}k \sqrt{k^2}(-\frac{e^{\rmi k|x+y-2b|}}{2\rmi k})\dd k
\\
=&-\frac{\rmi}{\pi}\int_{-\tilde{\Gamma}}k \sqrt{k^2}(-\frac{e^{\rmi k|x+y-2b|}}{2\rmi k})\dd k
\\
=&-\frac{2\rmi}{\pi}\int_0^{\infty}(\rmi k)^2(-\frac{e^{\rmi(\rmi k)|x+y-2b|}}{2\rmi(\rmi k)})\dd \rmi k
\\
=&-\frac{1}{\pi(x+y-2b)^2} .
\end{align*}


The same calculation yields for $x,y<b$. That is 
\begin{equation*}
[\breve{H}_\free-\breve{H}_{(b,+\infty)}](x,y)=[\breve{H}_\free-\breve{H}_{(-\infty,b)}](x,y)=-\frac{1}{\pi(x+y-2b)^2} .
\end{equation*}
When restricting to the diagonal, we have
\begin{equation}
\label{1D Delrel left right}
[\breve{H}_\free-\breve{H}_{(b,+\infty)}](x,x)=[\breve{H}_\free-\breve{H}_{(-\infty,b)}](x,x)=-\frac{1}{4\pi(x-b)^2} .
\end{equation}

Now for $a<x,y<b$, we have
\begin{equation*}
[G_{(a,b)}-G_\free](x,y;k^2)=\frac{\cos(k(b-a-|x-y|))-\cos(k(x+y-b-a))}{2k\sin (k(b-a))}+\frac{e^{\rmi k|x-y|}}{2\rmi k} ,
\end{equation*}
which implies
\begin{align*}
[\breve{H}_{(a,b)}-\breve{H}_\free](x,y)=&\frac{\rmi}{\pi}\int_{\tilde{\Gamma}}k \sqrt{k^2}[G_{(a,b)}-G_\free](x,y;k^2) \dd k
\\
=&-\frac{\rmi}{\pi}\int_{-\tilde{\Gamma}}k \sqrt{k^2}[G_{(a,b)}-G_\free](x,y;k^2) \dd k
\\
=&-\frac{2\rmi}{\pi}\int_0^{\infty}(\rmi k)^2[G_{(a,b)}-G_\free](x,y;(ik)^2) \dd \rmi k
\\
=&-\frac{2}{\pi}\int_0^{\infty}k^2[G_{(a,b)}-G_\free](x,y;(ik)^2) \dd k.
\end{align*}

Note that 

\begin{equation*}
\int_0^{\infty}\frac{k \cosh(ak)}{\sinh(bk)}\dd k=\frac{\pi^2}{4b^2}\sec^2\left(\frac{a\pi}{2b}\right) \quad \text{for} \quad a<b
\end{equation*}
and
\begin{equation*}
[G_{(a,b)}-G_\free](x,y;(\rmi k)^2)=\frac{\cosh(k(b-a-|x-y|))-\cosh(k(x+y-b-a))}{2k\sinh (k(b-a))}-\frac{e^{-k|x-y|}}{2k}
\end{equation*}
implies
\begin{align*}
&[\breve{H}_{(a,b)}-\breve{H}_\free](x,y)
\\
=&-\frac{2}{\pi}\int_0^{\infty} k^2[G_{(a,b)}-G_\free](x,y;(\rmi k)^2) \dd k
\\
=&-\frac{\pi}{4(b-a)^2}\left[
\csc^2\left(\frac{|x-y|\pi}{2(b-a)}\right)-\csc^2\left(\frac{(x+y-2b)\pi}{2(b-a)}\right)
\right]+\frac{1}{\pi(x-y)^2} .
\end{align*}
When restricting to the diagonal, we have
\begin{equation}
\label{1D Delrel mid}
[\breve{H}_{(a,b)}-\breve{H}_\free](x,x)=-\frac{\pi}{12(b-a)^2}+\frac{\pi}{4(b-a)^2}\csc^2\left(\frac{(x-b)\pi}{b-a}\right) .
\end{equation}
Equation \eqref{1D Grel} gives
\begin{equation*}
\frac{1}{2}H_\rel (x)=
\left\{
\begin{aligned}
&\frac{1}{2}[\breve{H}_\free-\breve{H}_{(-\infty,a_2)}](x,x)& x<b_1\\
&\frac{1}{2}[\breve{H}_{(b_1,a_2)} -\breve{H}_{(b_1,\infty)} -\breve{H}_{(-\infty,a_2)} + \breve{H}_\free](x,x)& b_1<x<a_2\\
&\frac{1}{2}[\breve{H}_\free-\breve{H}_{(-\infty,b_1)}](x,x)& x>a_2
\end{aligned}
\right. .
\end{equation*}
From equations \eqref{1D Delrel left right} and \eqref{1D Delrel mid}, we have
\begin{multline*}
\frac{1}{2}H_\rel (x)=
\\
\left\{
\begin{aligned}
&-\frac{1}{8\pi(x-a_2)^2}& x<b_1\\
&\text{\small{$\frac{\pi}{8(a_2-b_1)^2}\csc^2\Big(\frac{(x-a_2)\pi}{a_2-b_1}\Big)-\frac{\pi}{24(a_2-b_1)^2}-\frac{1}{8\pi(x-a_2)^2}-\frac{1}{8\pi(x-b_1)^2}$}}& b_1<x<a_2\\
&-\frac{1}{8\pi(x-b_1)^2}& x>a_2
\end{aligned}
\right. .
\end{multline*}
This equation shows that $H_\rel (x)$ is continuous, which is consistent with the claim in the proof of Theorem \ref{variation of energy thm}. Integrating over $\R$, we have
\begin{equation}
\label{1D Tr Hrel 0}
\frac{1}{2}\Tr_{\R}(H_\rel)=-\frac{1}{8\pi(a_2-b_1)}+\frac{6-\pi^2}{24\pi(a_2-b_1)}-\frac{1}{8\pi(a_2-b_1)}=-\frac{\pi}{24(a_2-b_1)} .
\end{equation}
Similarly, one has the renormalised counterpart of $H_\rel (x)$, which is given by $(H_\calO)_{\text{ren}}=(H_{\calO} -H_\free)|_\diag$. Note that this only corresponds to the first term in $T_{00}$ of \eqref{stress energy tensor}, i.e. $ \frac{1}{2}(H-H_\free)|_\diag$. It is given by
\begin{equation*}
\frac{1}{2}H_{\text{ren}} (x)=
\left\{
\begin{aligned}
&\frac{1}{8\pi(x-a_1)^2}& x<a_1\\
&\text{\small{$\frac{\pi}{8(b_1-a_1)^2}\csc^2\Big(\frac{(x-b_1)\pi}{b_1-a_1}\Big)-\frac{\pi}{24(a_2-b_1)^2}$}}& a_1<x<b_1\\
&\text{\small{$\frac{\pi}{8(a_2-b_1)^2}\csc^2\Big(\frac{(x-a_2)\pi}{a_2-b_1}\Big)-\frac{\pi}{24(a_2-b_1)^2}$}}& b_1<x<a_2\\
&\text{\small{$\frac{\pi}{8(b_2-a_2)^2}\csc^2\Big(\frac{(x-b_2)\pi}{b_2-a_2}\Big)-\frac{\pi}{24(b_2-a_2)^2}$}}& a_2<x<b_2\\
&\frac{1}{8\pi(x-b_2)^2}& x>b_2
\end{aligned}
\right. .
\end{equation*}
It is easy to see that $H_{\text{ren}}$ is not integrable. Therefore, some regularisation schemes would be needed at this point. One way is by heat-kernel regularisation (see, for instance, \cite{fulling2007vacuum}). However, this only resolves the non-integrability problem of the first term of $T_{00}$. We also need to integrate the term $\frac{1}{8}\Delta[(H^{-1}-H_\free^{-1})|_\diag]$ in equation \eqref{stress energy tensor} over $\R$. This is also ill-defined, as it is not integrable. We will see that these problems disappear when we work in the relative setting.

Restricting equation \eqref{1D Grel} to the diagonal and then taking the action of Laplacian, we have
\begin{multline*}
\label{1D DelGrel}
\Delta(G_\rel(x,x;(\rmi k)^2))=
\\
\left\{
\begin{aligned}
&2ke^{-2k(a_2-x)}& x<b_1\\
&-2k\frac{\cosh(k(2x-a_2-b_1))}{\sinh(k(a_2-b_1))}+2k[e^{-2k(a_2-x)}+e^{-2k(x-b_1)}]& b_1<x<a_2\\
&2ke^{-2k(x-b_1)}& x>a_2
\end{aligned}
\right. .
\end{multline*}
Integrating spectral variable $k$ along $\tilde{\Gamma}$ and then over the space variable $x$, we have 
\begin{equation}
\label{1Dsecondterm}
\int_\R \Delta(H_\rel^{-1}|_\diag) \dd x=0,
\end{equation}
hence
\begin{equation}
\label{1D Tr Hrel 1}
E_\rel=\frac{1}{2}\Tr_{\R}(H_\rel)+\int_\R\frac{1}{8}\Delta(H_\rel^{-1}|_\diag) \dd x=-\frac{\pi}{24(a_2-b_1)} .
\end{equation}
Note that using heat-kernel regularisation, one would also obtain $E_\rel$, see \cite{fulling2007vacuum}. Equations \eqref{1D Tr Hrel 0} and \eqref{1D Tr Hrel 1} agree with Theorem \ref{Erel thm}. Note that equation \eqref{1Dsecondterm} also shows that
\begin{equation*}
\int_\R \Delta[(H_\calO^{-1}-H_\free^{-1})|_\diag] \dd x \text{``$=$''} \int_\R \Delta[(H_{\calO_1}^{-1}-H_\free^{-1})|_\diag] \dd x+\int_\R \Delta[(H_{\calO_2}^{-1}-H_\free^{-1})|_\diag] \dd x,
\end{equation*}
where all the three terms are ill-defined as they are not integrable. For instance, $\Delta[(H_{\calO_1}^{-1}-H_\free^{-1})|_\diag](x)$ has $\frac{1}{(a_1-x)^2}$ singularity when approaching $x=a_1$. This justifies Remark \ref{divremark}.

Now let $X$ be a smooth vector field that generates a movement of (right) obstacle 2 to the right with a constant speed $v$. Moreover, $X$ is zero around (left) obstacle 1. In other words, we move the obstacle 2 to the right by $v\epsilon$ and keep the obstacle 1 stationary. Now the variation of the relative energy is given by $\delta_X E=v\cdot \partial_{a_2} E$. The left hand side of equation in Theorem \ref{variation of energy thm eqn main} becomes
\begin{equation}
\label{1D Erel}
\delta_X E_\rel=\frac{v\pi}{24(a_2-b_1)^2}.
\end{equation}
Now the identity \eqref{Tr delta Hrel} used in the proof of Theorem \ref{Equivalent thm} says
\begin{equation*}
\frac{1}{2}\Tr[\delta_X H_\rel]=-\frac{v}{4}\sum_{p=1}^{N}\int_{\partial \calO_p} [\partial_\nu\partial_\nu'(H_\calO^{-1}-H_{\calO_p}^{-1})]|_\diag  \dd x .
\end{equation*}
It becomes
\begin{equation}
\label{1D Tr Hrel 2}
\begin{aligned}
\frac{1}{2}\Tr[ \delta_X H_\rel]=&-\frac{v}{4}\left[ [\partial_\nu\partial_\nu'(H_\calO^{-1}-H_{\calO_2}^{-1})](a_2,a_2)-[\partial_\nu\partial_\nu'(H_\calO^{-1}-H_{\calO_2}^{-1})](b_2,b_2) \right]
\\
=&-\frac{v}{4}[\partial_\nu\partial_\nu'(H_\calO^{-1}-H_{\calO_2}^{-1})](a_2,a_2).
\end{aligned}
\end{equation}
Note that
\begin{equation*}
\partial_x\partial_y\left(G_{(a,b)}-G_{(-\infty,b)}\right)(b,b;k^2)=-\rmi k-k\cot(k(b-a)) ,
\end{equation*}
therefore
\begin{equation}
\label{1D Tr Hrel 3}
\begin{aligned}
&[\partial_\nu\partial_\nu'(H_\calO^{-1}-H_{\calO_2}^{-1})](a_2,a_2)
\\
=&\frac{\rmi}{\pi}\int_{\tilde{\Gamma}}\frac{k}{\sqrt{k^2}}\partial_x\partial_y[G_{(b_1,a_2)}-G_{(-\infty,a_2)}](a_2,a_2;k^2) \dd k
\\
=&-\frac{\rmi}{\pi}\int_{-\tilde{\Gamma}}\frac{k}{\sqrt{k^2}}\partial_x\partial_y[G_{(b_1,a_2)}-G_{(-\infty,a_2)}](a_2,a_2;k^2) \dd k
\\
=&-\frac{2\rmi}{\pi}\int_0^{\infty}\partial_x\partial_y[G_{(b_1,a_2)}-G_{(-\infty,a_2)}](a_2,a_2;(\rmi k)^2) \rmi \dd k
\\
=&\frac{2}{\pi}\int_0^{\infty}(k-k\coth(k(a_2-b_1))) \dd k
\\
=&-\frac{\pi}{6(a_2-b_1)^2}.
\end{aligned}
\end{equation}
Combining equations \eqref{1D Tr Hrel 1}, \eqref{1D Tr Hrel 2} and \eqref{1D Tr Hrel 3}, we have verified the identity \eqref{Tr delta Hrel} in one dimensional cases. Moreover, equations \eqref{1D Erel}, \eqref{1D Tr Hrel 2} and \eqref{1D Tr Hrel 3} are consistent with Theorem \ref{variation of energy thm} and Theorem \ref{Equivalent thm}.

\end{document}